\documentclass[10pt]{article}
\usepackage[margin=1in]{geometry}
\usepackage{amsmath,amsthm,amssymb,eurosym}
\usepackage{enumerate}
\usepackage{graphicx}
\usepackage{fancyhdr}
\usepackage{lastpage}
\usepackage{titlesec}
\usepackage[OT2,T1]{fontenc} 
\usepackage[russian,american]{babel} 
\usepackage{lastpage}
\usepackage{float}
\usepackage{bbm}
\usepackage[usenames,dvipsnames]{xcolor}
\usepackage{color}
\usepackage{wrapfig}
\usepackage{verbatim}
\usepackage{algorithm}
\usepackage{algorithmic}
\usepackage{framed} 

\usepackage{url}
\usepackage{booktabs}
\usepackage{multirow}
\usepackage{subcaption}
\usepackage{hyperref}

\usepackage[utf8]{inputenc}
\usepackage{csquotes} 

\newcommand\hl[1]{%
  \bgroup
  \hskip0pt\color{red}%
  #1%
  \egroup
}
\renewcommand{\hl}{}

\newcommand{\R}{\mathbb{R}}   

\renewcommand{\P}{\mathbf{P}}   
\newcommand{\E}{\mathbf{E}}   
\newcommand{\var}{\operatorname{Var}}   
\newcommand{\cov}{\operatorname{Cov}}   

\newcommand{\cY}{\mathcal{Y}}

\newcommand{\cN}{\mathcal{N}}

\newcommand{\cW}{\mathcal{W}}

\newcommand{\one}{\mathbbm{1}}  

\newcommand{\+}[1]{\mathbf{#1}}

\newcommand{\ep}{\varepsilon}

\newcommand{\on}[1]{\operatorname{#1}}

\newcommand{\pto}{\stackrel{p}{\to}}
\newcommand{\wto}{\Rightarrow}

\newcommand{\iid}{\stackrel{\text{iid}}{\sim}} 
\theoremstyle{plain}
\newtheorem{theorem}{Theorem}
\newtheorem{corollary}{Corollary}

\newtheorem{lemma}{Lemma}
\newtheorem{proposition}{Proposition}
\newtheorem{assumption}{Assumption}

\newtheorem{model}{Model}
\theoremstyle{definition}

\usepackage{natbib}
\setcitestyle{authoryear,aysep={}}

\makeatletter
\newenvironment{subtheorem}[1]{%
  \def\subtheoremcounter{#1}%
  \refstepcounter{#1}%
  \protected@edef\theparentnumber{\csname the#1\endcsname}%
  \setcounter{parentnumber}{\value{#1}}%
  \setcounter{#1}{0}%
  \expandafter\def\csname the#1\endcsname{\theparentnumber(\alph{#1})}%
  \ignorespaces
}{%
  \setcounter{\subtheoremcounter}{\value{parentnumber}}%
  \ignorespacesafterend
}
\makeatother
\newcounter{parentnumber}

\newcommand{\indep}{\perp\!\!\!\perp}

\newcommand{\HT}{\hat \tau_{\text{HT}}}
\newcommand{\hajek}{\hat \tau_{\text{H\'ajek}}}
\newcommand{\DM}{\hat \tau_{\text{DM}}}

\newcommand{\sumin}{\sum_{i=1}^n} 	

\newcommand{\avgin}{\frac{1}{n}\sumin}
\newcommand{\limn}{\lim_{n\to\infty}}

\newcommand{\Wni}{\+W_{-i}}

\title{
	Regression adjustments for estimating the global treatment effect in experiments with interference\thanks{The author thanks Fredrik S\"avje, Johan Ugander, and seminar participants at Facebook, Stanford University, and Yale University 
  for helpful comments and suggestions.  
  This work was supported in part by NSF grant IIS-1657104.} 
}
\author{
	Alex Chin\thanks{Department of Statistics, Stanford University, Stanford, CA, 94305 USA (\texttt{ajchin@stanford.edu})}
}
\date{This version: \today}

\begin{document}
\maketitle

\begin{abstract}
Standard estimators of the global average treatment effect can be biased in the presence of interference.  This paper proposes regression adjustment estimators for removing bias due to interference in Bernoulli randomized experiments.  We use a fitted model to predict the counterfactual outcomes of global control and global treatment.  Our work differs from standard regression adjustments in that the adjustment variables are constructed from functions of the treatment assignment vector, and that we allow the researcher to use a collection of any functions correlated with the response, turning the problem of detecting interference into a feature engineering problem.  We characterize the distribution of the proposed estimator in a linear model setting and connect the results to the standard theory of regression adjustments under SUTVA.  We then propose an estimator that allows for flexible machine learning estimators to be used for fitting a nonlinear interference functional form.  We propose conducting statistical inference via bootstrap and resampling methods, which allow us to sidestep the complicated dependences implied by interference and instead rely on empirical covariance structures.  Such variance estimation relies on an exogeneity assumption akin to the standard unconfoundedness assumption invoked in observational studies.  In simulation experiments, our methods are better at debiasing estimates than existing inverse propensity weighted estimators based on neighborhood exposure modeling.  We use our method to reanalyze an experiment concerning weather insurance adoption conducted on a collection of villages in rural China.
\\[\baselineskip]
\end{abstract}

\section{Introduction}

The goal in a randomized experiment is often to estimate the \emph{total} or \emph{global average treatment effect} (GATE) of a binary \hl{treatment} variable on a response variable.  The GATE is the difference in average outcomes when all units are exposed to treatment versus when all units are exposed to control. 
Under the standard assumption that units do not interfere with each other~\citep{cox1958planning}, which forms a key part of the \emph{stable unit treatment value assumption} (SUTVA)~\citep{rubin1974estimating,rubin1980randomization}, the global average treatment effect reduces to the standard average treatment effect.


However, in many social, medical, and online settings the no-interference assumption may fail to hold~\citep{rosenbaum2007interference,walker2014design,aral2016networked,taylor2017randomized}.  In such settings, peer and spillover effects can bias estimates of the global average treatment effect.  In the past decade, there has been a flurry of literature proposing methods for handling interference, mostly focusing on cases in which structural assumptions about the nature of interference are known.  For example, if there is a natural grouping structure to the data, such as households or schools or classrooms, it may be reasonable to assume that interference exists within but not across groups.  Versions of this assumption are known as \emph{partial} or \emph{stratified interference}~\citep{hudgens2008toward}.  In this case two-stage randomized designs can be used to decompose direct and indirect effects, which is an approach studied by~\citet{vanderweele2011effect,tchetgen2012causal,liu2014large,baird2016optimal,basse2017exact}, among others.  \citet{baird2016optimal} study how two-stage, random saturation designs can be used to estimate dose response curves under the stratified interference assumption.  \citet{basse2018analyzing} study two-stage experiments in which households with multiple students are assigned to treatment or control.  Other works that propose methods of handling interference include~\cite{ogburn2014causal}, which maps out causal diagrams for interference; \citet{van2014causal}, which studies a \hl{targeted maximum likelihood estimator for the case where network connections and treatments possibly change over time}; \citet{choi2017estimation}, which shows how confidence intervals can be constructed in the presence of monotone treatment effects; and~\citet{jagadeesan2017designs}, which studies designs for estimating the direct effect that strive to balance the network degrees of treated and control units.

The \emph{modus operandi} for general or arbitrary interference is the method of \emph{exposure modeling}, in which the researcher defines equivalence classes of treatments that inform the interference pattern.  \citet{aronow2017estimating} develop a general framework for analyzing inverse propensity weighted (Horvitz-Thompson- and H\'ajek-style) estimators under correct specification of local exposure models.  The exposure model often  used is some version of an  assumption that the potential outcomes of unit $i$ are constant conditional on all treatments in a local neighborhood of $i$, or that the potential outcomes are a monotone function of such treatments.  This assumption, known as \emph{neighborhood treatment response} (NTR), is a generalization of partial and stratified interference to the general network setting~\citep{manski2013identification}.  Methods for handling interference often rely on neighborhood treatment response as a core assumption.  For example, \citet{sussman2017elements} develop unbiased estimators for various parametric models of interference that are all restrictions on the NTR condition, and~\citet{forastiere2016identification} propose propensity score estimators for observational studies using the NTR assumption.

\citet{aronow2017estimating} use their methods to analyze the results of a field experiment on an anti-conflict program in middle schools in New Jersey.  By defining appropriate exposure models, they are able to estimate a direct effect (the effect of receiving the anti-conflict intervention), a spillover effect (the effect of being friends with some students who received the anti-conflict intervention), and a school effect (the effect of attending a school in which some students received the anti-conflict intervention).  The network structure consists of 56 disjoint social networks (schools), comprising 24,191 students in the original~\citet{paluck2016changing} study and a subset of 2,050 students studied in the~\citet{aronow2017estimating} analysis.  There are a number of similar studies in which the target of scientific inquiry is the quantification of peer or spillover effects and where the dataset permits doing so by being comprised of ``many sparse networks.''  Studies which consist of randomized experiments on such social networks include \citet{banerjee2013diffusion}, which studies a microfinance loan program in villages in India; \citet{cai2015social}, which studies a weather insurance program for farmers in rural China; \citet{kim2015social}, which concerns public health interventions such as water purification and microvitamin tablets in villages in Honduras; and~\citet{beaman2018can}, which explores social diffusion of a new agricultural technology among farmers in Malawi.  (Some studies thereof do not explicitly aim to understand spillover effects---for example~\citet{kim2015social} and ~\citet{beaman2018can} are concerned primarily with strategies for targeting influential individuals---but the presence of such effects is still crucial for their purposes.)  In these settings, exposure modeling may be (and has been) a successful way of decomposing direct and spillover effects.

\paragraph{\hl{The difficulties of using exposure models for global effects}}

How should one proceed if the goal is estimation of the global treatment effect rather than a decomposition into direct and spillover effects?  In this setting interference is a nuisance, not an object of intrinsic scientific interest.  \hl{Unbiased estimation would result from using an exposure model that accurately represents the true data-generating process.  However, the complicated nature of social interactions makes it difficult to select an exposure model that is both tractable and well-specified.} \citet{eckles2017design} discuss some of the difficulties of working in this setting \hl{in the context of  ``implausibility of tractable treatment response assumptions'':}
\begin{quote}
It is unclear how substantive judgment can directly inform the selection of an exposure model for interference in networks---at least when the vast majority of vertices are in a single connected component.  Interference is often expected because of social interactions (i.e., peer effects) where vertices respond to their neighbors' behaviors: in discrete time, the behavior of a vertex at $t$ is affected by the behavior of its neighbors at $t - 1$; if this is the case, then the behavior of a vertex at $t$ would also be affected by the behavior of its neighbors' neighbors at $t - 2$, and so forth.  Such a process will result in violations of the NTR assumption, and many other assumptions that would make analysis tractable.
\end{quote}

In this setting, one primary tool that has developed in the literature is the method of \emph{graph cluster randomization}~\citep{ugander2013graph}, where researchers use a clustered design in which the clusters are selected according to the structure of the graph in order to lower the variance of NTR-based inverse propensity estimators.  \citet{eckles2017design} provide theoretical results and simulation experiments to show how clustered designs can reduce bias due to interference.  Clusters can be obtained using algorithms developed in the graph partitioning and community detection literature~\citep{fortunato2010community,ugander2013balanced}.

While the graph clustering approach can be effective at removing some bias, the structure of real-world empirical networks may make it difficult to obtain satisfactory bias reduction via clustering, which relies on having good quality graph cuts.  The ``six degrees of separation'' phenomenon is well-documented in large social networks~\citep{ugander2011anatomy,backstrom2012four}, and the average distance between two Facebook users in February 2016 was just 3.5~\citep{bhagat16}.  Furthermore, most users belong to one large connected component and are unlikely to separate cleanly into evenly-sized clusters.  In a graph clustered experiment run at LinkedIn, the optimal clustering strategy used maintained only 35.59\% of edges between nodes of the same cluster~\citep[Table 1]{saveski2017detecting}, suggesting that bias remains even after clustering.  Figure~\ref{fig:schematic} provides an example illustration of how the structure of the network can markedly affect how much we might expect cluster randomization to help.


\begin{figure}[ht]
\centering
\includegraphics[width=0.9\textwidth]{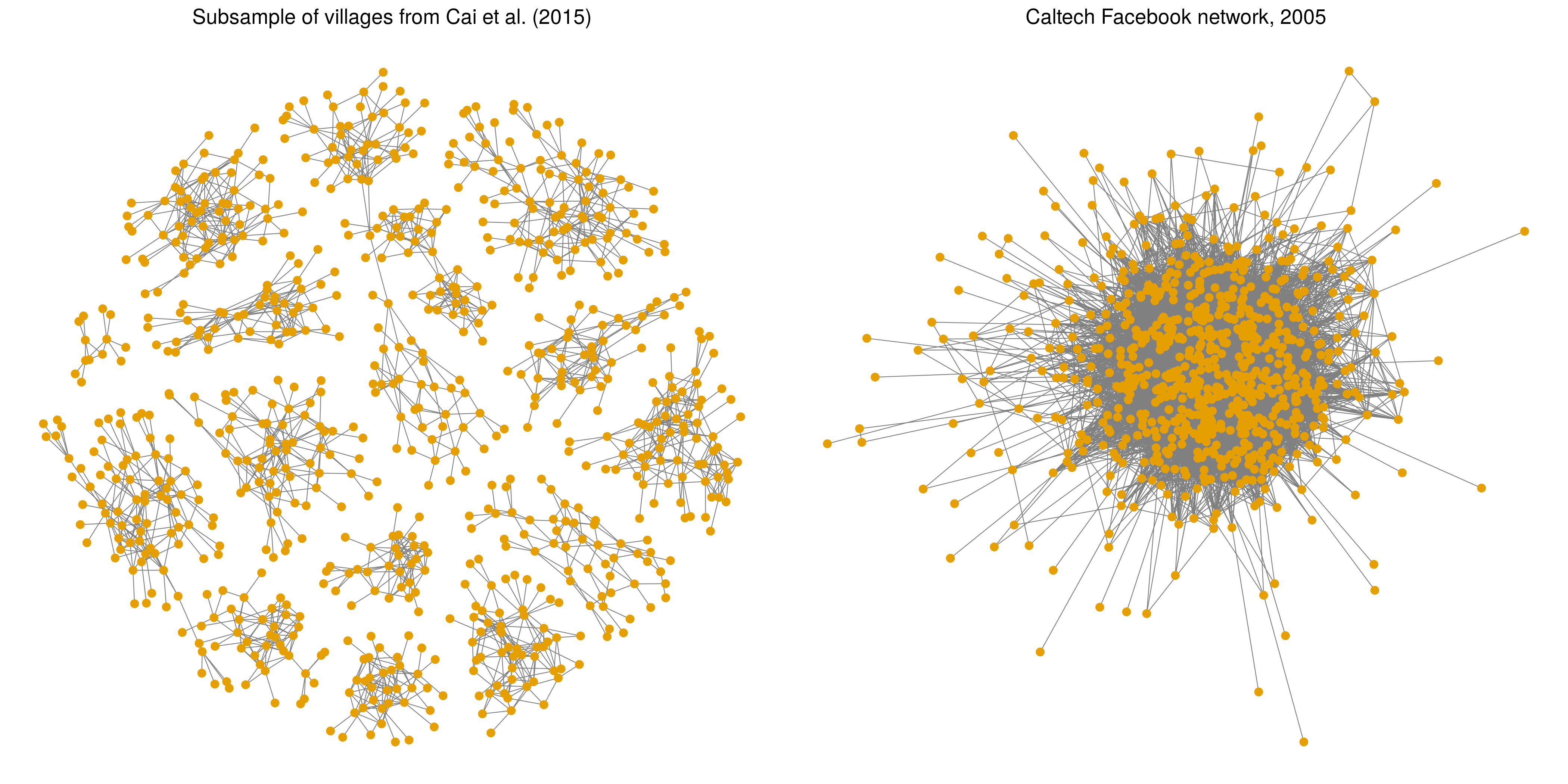}
\caption{(left) A subset of 16 nearly-disjoint Chinese villages, comprising 822 nodes, from an experiment regarding weather insurance adoption conducted by~\citet{cai2015social}. The setup of many, sparse networks is similar to that in the anti-conflict school dataset from~\citet{paluck2016changing}.  (right) The largest connected component of the Caltech Facebook network, with 762 nodes, from a single day snapshot in September 2005, taken from the \texttt{facebook100} dataset~\citep{traud2011comparing,traud2012social}.  Networks were plotted with the \texttt{ggnet2} function~\citep{tyner2017network} in the \texttt{GGally} package, using the default Fruchterman-Reingold force-directed layout~\citep{fruchterman1991graph}. 
We should not be surprised if methods for handling interference that might work well in the collection of networks on the left, such as exposure modeling and graph clustering, do not work so well in the network on the right.}
\label{fig:schematic}
\end{figure}

Such experimental designs also face practical hurdles.  Though cluster randomized controlled trials are commonly used in science and medicine, existing experimentation platform infrastructure in some organizations may only exist for standard (i.i.d.) randomized experiments, in which case adapting the design and analysis pipelines for graph cluster randomization would require significant ad hoc engineering effort.  In regimes of only mild interference, it may simply not be worth the trouble to run a clustered or two-stage experiment, especially if there is no way to know \emph{a priori} how much bias from interference will be present.  Instead, the practitioner would prefer to have a data-adaptive debiasing mechanism that can be applied to an experiment that has already been run.  Ideally, such estimators provide robustness to deviations from SUTVA yet do not sacrifice too much in precision loss if it turns out interference was weak or non-existent.

\paragraph{\hl{Towards an agnostic regression approach}}
\hl{We can take advantage of the fact that the global treatment effect estimand, as opposed to a peer or spillover effect estimand, can be defined \emph{agnostically} without regard to any exposure model.  The exposure model used, therefore, matters only insofar as it informs the corresponding estimator used.  An appropriate exposure model is one that leads to estimates of the global treatment effect that are approximately unbiased, \emph{even if it is not the exposure model corresponding to the true data-generating process}.  This agnostic perspective gives us hope because of the decoupling between data generation and estimation: We can believe in a complex interference pattern without having to use the corresponding intractable exposure model for estimation. }

Our approach is motivated by the rich literature on regression adjustment estimators in the non-interference setting.  In randomized controlled trials, regression adjustments are used to adjust for imbalances due to randomized assignment of the empirical covariate distributions of different treatment groups, and thus improve precision of treatment effect estimators.  In the observational studies setting, regression adjustments are used to adjust for inherent differences between the covariate distributions of different treatment groups.  We heavily borrow tools from that literature, both in the classical regime of using low-dimensional, linear regression estimators~\citep{freedman2008regressionA,freedman2008regressionB,lin2013agnostic,berk2013covariance} and more recent advancements that can utilize high-dimensional regression and machine learning techniques~\citep{bloniarz2016lasso,wager2016high,wu2017loop,athey2017approximate,chernozhukov2018double}.  \hl{This recent literature adopts the agnostic perspective that properties of least squares and machine learning estimators can be utilized without assuming the parametric model itself.}

This paper contains two main contributions: (a) a regression adjustment strategy for debiasing global treatment effect estimators, and (b) a class of bootstrapping and resampling methods for constructing variance estimates of such estimators.  We explore how well the analysis side of an experiment can be improved in independently-assigned (non-clustered) experiments.  Our approach can be loosely motivated by the \emph{linear-in-means} (LIM) family of models from the econometrics literature~\citep{manski1993identification}.  \hl{In a simple version of this model, an individual's outcome is said to depend on the average of her peer's exogenous features.  If this is true, then the peer average feature ``statistic'' can be adjusted for when estimating the global treatment effect, even if the linear-in-means model itself does not hold.  We note that much of the linear-in-means literature focuses on the identifiability of various peer effect parameters within the LIM model~\citep{bramoulle2009identification}; our goal instead is estimation of the agnostic global effect.}

Generally, our strategy is to learn a statistical model that captures the relationship between the outcomes and a set of unit-level statistics constructed from the treatment vector and the observed network.  These statistics can be viewed as \hl{\emph{features} or \emph{adjustment variables}}, and are to be constructed by the practitioner using domain knowledge.  The model is then used to predict the unobserved potential outcomes of each unit under the counterfactual scenarios if the unit had been assigned to global treatment, and global control.   The approach is thus reminiscent of regression adjustment estimators and off-policy evaluation.  Figure~\ref{fig:cf-distns} demonstrates how \hl{feature} distributions differ between the observed design distribution and the unobserved global counterfactual distributions of interest.

\hl{In Section~\ref{sec:linear} we present estimators in the context of a generative linear model and in Section~\ref{sec:nonparametric} we discuss the non-linear analog.  Even though the results in this paper are presented within the context of a generative model, they still make progress towards a fully agnostic solution.  First, the models considered here are considerably more flexible, and more easily extended, than exposure models (which are also assumed to be generative).  Second, the assumptions on the errors can be relaxed and we show via simulation experiments (Section~\ref{sec:simulations}) that these linear methods can work well in more general contexts.  The non-linear context, which allows the use of arbitrary machine learning estimators, also moves closer to a fully agnostic approach by allowing a nonparametric generative model.  Third, by connecting these results to analogous ones in the SUTVA case we lay the foundation for how to think about an agnostic approach, which is not possible using pure exposure modeling methods.  Indeed, agnostic perspectives have emerged only recently even in the SUTVA setting~\citep{freedman2008regressionA,freedman2008regressionB,lin2013agnostic}. This paper, therefore, can be viewed as a conceptual stepping stone between existing methods that assume exposure models are generative, and future work that would establish a fully agnostic presentation.}

\begin{figure}[t]
\centering
\includegraphics[width=0.9\textwidth]{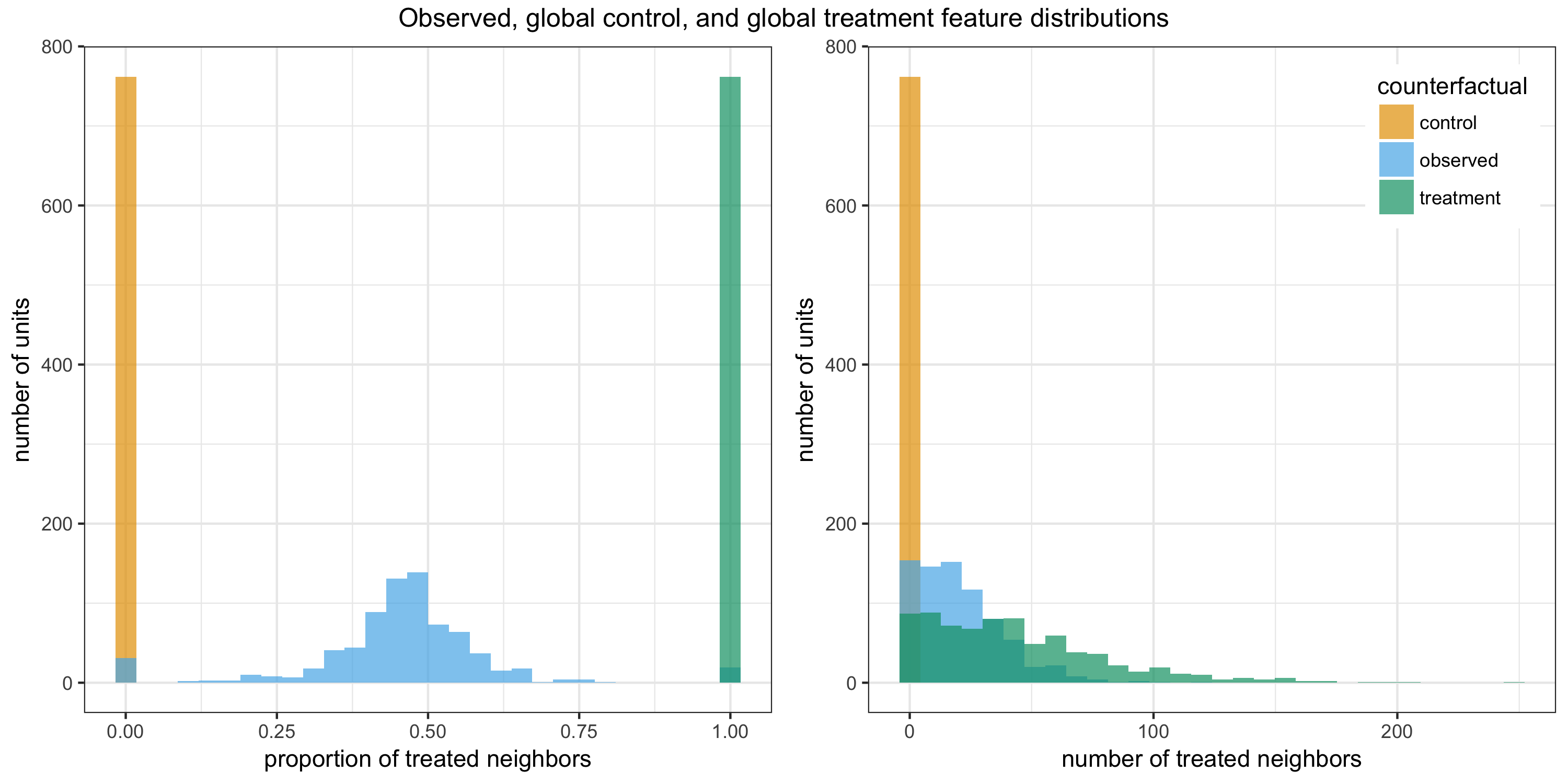}
\caption{(left) Distributions for fraction of treated neighbors $d_i^{-1} \sum_{j \in \cN_i} W_j$.  (right) Distributions for number of treated neighbors $\sum_{j \in \cN_i} W_j$.  Feature distributions are under global exposure to control $\+W = \+0$ (orange), global exposure to treatment $\+W = \+1$ (green), and a single observed treatment instance from an iid Bernoulli$(0.5)$ distribution (blue).  Network is the Caltech social graph from the \texttt{facebook100} dataset~\citep{traud2011comparing,traud2012social}.  If the response is correlated with one or both of these features, then ideas from off-policy evaluation of the counterfactual outcomes can guide estimation of the global treatment effect.  Even if the distributions are quite different, as in the left hand picture, if the response can be modeled by low dimensional model then extrapolation may not be too unreasonable.}
\label{fig:cf-distns}
\end{figure}

\hl{It is shown in this paper that an assumption of exogeneity is required, even though the treatment is randomized.  Such an assumption can be likened to an \emph{unconfoundedness}, \emph{ignorability}, or \emph{selection on observables} assumption.  A curious feature of randomized experiments under interference, then, is that they display characteristics of observational studies as well.  It is helpful to think of estimators used in SUTVA observational studies that require the estimation of both a \emph{propensity model} and a \emph{response model}.  (Doubly-robust estimators allow misspecification of one but not both of these models.)}  In a randomized experiment under interference the propensity model is fully known and does not need to be estimated; however, the response can be affected by confounding variables.  In randomized experiments under interference, then, researchers must be wary of the same challenges that beset drawing causal conclusions from observational datasets, even though the treatments were assigned randomly.  The exogeneity assumption is not generally verifiable from the data but is necessary in order to make any progress.  Ideally, one has access to methods for conducting sensitivity analyses for interference, but such methods are in their infancy and we refrain from addressing this issue here.

Our estimators have several advantages over existing exposure modeling estimators.  The correct specification of an exposure model is also a form of exogeneity assumption, yet our approach admits much more flexible forms of interference.  It can handle multiple types of graph features, which do not even have to be constructed from the same network.  Adjusting for interference becomes a feature engineering problem in which the practitioner is free to use his or her domain knowledge to construct appropriate features.  If a feature turns out to be noninformative for interference, no additional bias is incurred (though a penalty in variance may be paid).  Our adjustment framework also reduces to the standard, SUTVA regression adjustment setup in the event that static, \hl{baseline characteristics} are used.

Finally, we propose methods for quantifying the variance of the proposed estimators.  Variance estimation in the presence of interference is generally difficult because of the complicated dependencies created by the propagation of interference over the network structure.  Confidence intervals based on asymptotic approximations may not be reliable since the dependencies can drastically reduce the effective sample size.  For example, the variance of the sample mean of the \hl{features} may not even scale at a $n^{-1}$ rate, where $n$ is the sample size.  In this paper we propose a novel way of taking advantage of the randomization distribution to produce bootstrap standard errors, assuming unconfoundedness.  Since the \hl{features} are constructed by the researcher from the vector of treatments, and the distribution of treatments is known completely in a randomized experiment, we can calculate via Monte Carlo simulation the sampling distribution of any function of the design matrix under the randomization distribution.  This approach ensures that we properly represent all of the dependencies exhibited empirically by the data, and can then be used to construct standard errors. 





The remainder of this paper is structured as follows.  In Section~\ref{sec:setup} we describe the problem and motivate our approach with an informal discussion of a linear-in-means model.  In Section~\ref{sec:linear} we develop the main results for linear regression estimators and in Section~\ref{sec:nonparametric} we show how to extend this to the non-linear setting.  In Section~\ref{sec:simulations} we conduct simulation experiments, in Section~\ref{sec:applications} we consider an application to an existing field experiment, and in Section~\ref{sec:discussion} we conclude.  All proofs are in the appendix.

\section{Setup and estimation in LIM models}
\label{sec:setup}
We work within the potential outcomes framework, or Rubin causal model~\citep{neyman1923application,rubin1974estimating}.  Consider a population of $n$ units indexed on the set $[n] = \{1, \dots, n\}$ and let $\+W = (W_1, \dots, W_n) \in \cW = \{0, 1\}^n$ be a random vector of binary treatments.  We will work only with treatments assigned according to a Bernoulli randomized experimental design:
\begin{assumption}
\label{asm:bernoulli-design}
$W_i \iid \on{Bernoulli}(\pi)$ for every unit $i \in [n]$, where $\pi \in (0, 1)$ is the treatment assignment probability.
\end{assumption}
The general spirit of our approach can likely be extended to more complicated designs, but our goal in this paper is to show that substantial analysis-side improvements can be made even under the simplest possible experimental design.

Suppose that each response lives in an outcome space $\cY$, and is determined by a mean function $\mu_i: \cW \to \cY$:  
\begin{equation}
\label{eqn:mean-function}
Y_i = Y_i(\+W) = \mu_i(\+W) + \ep_i
\end{equation}
In this section we limit ourselves to an informal discussion of point estimation and defer the question of variance estimation to a future section. The only assumption we require on the residuals, therefore, is an assumption of strict exogeneity:
\[\E[\ep_i | W_1, \dots, W_n] = 0.\]
In particular, no independence or other assumptions about the correlational structure of the residuals are made in this section, though such assumptions will be necessary for variance estimation, which we address in Section~\ref{sec:linear}.

Because the units are assumed to belong to a network structure, distinguishing between finite population and infinite superpopulation setups is not so straightforward.  In the SUTVA setting, good estimators for finite population estimands (or conditional average treatment effects) are usually good estimators for superpopulation estimands, and vice versa~\citep{imbens2004nonparametric}.  In order to simplify the analysis, we do not work with a fixed potential outcomes $Y_i(\+w)$ for $\+w \in \cW$, and allow the residuals $\ep_i$ to be random variables.  We therefore consider additional variation of the potential outcomes coming from repetitions of the experiment, but we do not consider the units to be sampled from a larger population.  We do this because it is easier to discuss the behavior of $\ep_i$ when they are random variables.  \hl{This perspective is related to the \emph{intrinsic nondeterminism} perspective discussed by \citet{pearl2009causality} on page 220, 
as well as the idea of \emph{stochastic counterfactuals} discussed previously in the literature~\citep{greenland1987interpretation,robins1989probability,robins2000causal,vanderweele2012stochastic}.}

In this paper we focus on estimation of the \emph{total} or \emph{global average treatment effect} (GATE), defined by
\begin{equation}
\label{eqn:tau-global}
\tau = \avgin [\E[Y_i(\+1)]  - \E[Y_i(\+0)]].
\end{equation}
This parameter is called a global treatment effect because is a contrast of average outcomes between the cases when the units are globally exposed to treatment ($\+W = \+1$) and globally exposed to control ($\+W = \+0$).

Under an assumption of strict exogeneity, in which $\E[\ep_i | \+W] = 0$, the treatment effect is the difference of average global exposure means
\[\tau = \avgin \left[\mu_i(\+1) - \mu_i(\+0)\right],\]
In order to proceed, we must make assumptions about the structure of the mean function $\mu_i$.

\subsection{A simple linear-in-means model}
To illustrate our approach we start with a simple model.  Let $G$ be a network with adjacency matrix $A$.  For simplicity in this paper we will mostly assume that $G$ is simple and undirected, but one can just as easily use a weighted and directed graph.  We emphasize that we assume $G$ is completely known to the researcher.  Let $\cN_i = \{j \in [n] : A_{ij} = 1\}$ be the neighborhood of unit $i$ and $d_i = |\cN_i|$ be the network degree of unit $i$.  Define
\begin{equation}
\label{eqn:frac-nbh}
X_i = \frac{1}{d_i} \sum_{j \in \cN_i} W_j,
\end{equation}
the fraction of neighbors of $i$ that are in the treatment group.  Then take the mean function $\mu_i$ in equation~\eqref{eqn:mean-function} to be as follows.
\begin{model}[Exogenous LIM model]
\[\mu_i(\+W) = \alpha + \gamma W_i + \delta X_i.\]
\label{model:lim-exo}
\end{model}
This model is a simple version of a linear-in-means model~\citep{manski1993identification}.  The model contains an intercept $\alpha$ as well as a direct effect $\gamma$, which captures the strength of individual $i$'s response to changes in its own treatment assignment.  Additionally, the response of unit $i$ is correlated with mean treatment assignment of its neighbors; \citet{manski1993identification} calls $\delta$ an \emph{exogenous social effect}, because it captures the correlation of unit $i$'s response with the exogenous characteristics of its neighbors.  The interactions are also assumed to be ``anonymous'' in that the unit $i$ responds only to the mean neighborhood treatment assignment and not the identities of those treated neighbors.  In this model, unit $i$ responds to its neighbors' treatments but not to its neighbors' outcomes.  \hl{Under Model~\ref{model:lim-exo}, the variable $X_i$ is the mechanism by which interference affects the outcome and thus can be viewed as playing a similar role as baseline characteristics or pretreatment covariates in an observational study.  However, in this paper we shall use the term \emph{statistic} or \emph{feature} rather than \emph{covariate} to refer to $X_i$, in order to remind the reader that $X_i$ does not represent a baseline characteristic.}

Now consider the estimand~\eqref{eqn:tau-global} under Model~\ref{model:lim-exo}.
If all units are globally exposed to treatment then it is the case for all units $i$ that $W_i = 1$ and $X_i = 1$.  Therefore
\[\avgin \mu_i{(\+1)} = \alpha + \gamma + \delta.\]
Similarly, if all units are globally exposed to control, then $W_i = 0$ and $X_i = 0$, and so
\[\avgin \mu_i{(\+0)} = \alpha.\]
Therefore, the treatment effect under Model~\ref{model:lim-exo} is simply
\[\tau = (\alpha + \gamma + \delta) - \alpha = \gamma + \delta.\]
This parametrization suggests that if we have access to unbiased estimators $\hat \gamma$ and $\hat \delta$ for $\gamma$ and $\delta$, then an unbiased estimate for $\tau$ is given by
\[\hat \tau = \hat \gamma + \hat \delta.\]
In particular, one is tempted to estimate $\gamma$ and $\delta$ with an OLS regression of $Y_i$ on $W_i$ and $X_i$.  Of course, using $\hat \tau$ as an estimator for $\tau$ only makes sense if Model~\ref{model:lim-exo} accurately represents the true data generating process.  We build up more flexible models in the following sections.

In contrast, we can easily see why the difference-in-means estimator,
defined for sample sizes $N_1 = \sumin W_i$ and $N_0 = \sumin (1 - W_i)$ as
\begin{equation}
\label{eqn:dm}
\DM = \frac{1}{N_1} \sumin W_iY_i - \frac{1}{N_0} \sumin (1 - W_i)Y_i,
\end{equation}
is biased under Model~\ref{model:lim-exo}.  The mean treated response is
\[\E[Y_i | W_i = 1] = \alpha + \gamma + \delta\E[X_i | W_i = 1] = \alpha + \gamma + \delta\E[X_i],\]
where $X_i$ is independent of $W_i$ since the treatments are assigned independently and there are no self-loops in $G$.  Similarly,
\[\E[Y_i | W_i = 0] = \alpha + \delta \E[X_i | W_i = 0] = \alpha + \delta \E[X_i].\]
Therefore, the difference-in-means estimator $\DM$ has expectation $\gamma$, which need not equal $\tau = \gamma + \delta$ in general.  Only if $\delta = 0$ do they coincide, in which case SUTVA holds and there is no interference.  In other words, the difference-in-means estimator marginalizes out the indirect effect rather than adjusting for it; it is an unbiased estimator not for the GATE but for the \emph{expected} average treatment effect (EATE), defined as
\[\avgin [\E[Y_i | W_i = 1] - \E[Y_i | W_i = 0]].\]
The EATE was introduced in~\citet{savje2017average} as a natural object of study for estimators which are designed for the SUTVA setting.   \citet{savje2017average,chin2018central} study the limiting behavior of estimators such as $\DM$ under mild regimes of misspecification of SUTVA due to interference.

\subsection{Linear-in-means with endogenous effects}
\label{sec:lim-endo}

Now we move to the more interesting version of the linear-in-means model, which contains an endogenous social effect in addition to an exogenous one.  Let
\begin{equation}
\label{eqn:endogenous-mean}
Z_i = \frac{1}{d_i} \sum_{j \in \cN_i} Y_j,
\end{equation}
the average value of the neighboring responses.  Now consider the following model:

\begin{subtheorem}{model}\label{model:lim-endo}
\begin{model}[LIM with endogenous social effect]
\label{model:lim-infinite}
\[\mu_i(\+W) = \alpha + \beta Z_i + \gamma W_i + \delta X_i.\]
\end{model}
In addition to direct and exogenous spillover effects, unit $i$ now depends on the outcomes of its neighbors through the spillover effect $\beta$.  It is conventional and reasonable to assume that $|\beta| < 1$.  Model~\ref{model:lim-infinite} is often more realistic than Model~\ref{model:lim-exo}; as discussed in the introduction, we often believe that interference is caused by individuals reacting to their peers' behaviors rather than to their peers' treatment assignments.  

It is helpful to write Model~\ref{model:lim-infinite} in vector-matrix form.  Let $\tilde G$ be the weighted graph defined by degree-normalizing the adjacency matrix of $G$; i.e., let $\tilde G$ be the graph corresponding to the adjacency matrix $\tilde A$ with entries $\tilde A_{ij} = d_i^{-1} A_{ij}$.  Then the matrix representation of Model~\ref{model:lim-infinite} is 
\begin{equation}
\label{eqn:lim-reduced}
Y = \alpha + \beta \tilde A Y + \gamma W + \delta \tilde A W + \ep,
\end{equation}
where $Y$, $W$, and $\ep$ are the $n$-vectors of responses, treatment assignments, and residuals, respectively.  
Using the matrix identity $(I - \beta \tilde A)^{-1} = \sum_{k=0}^\infty \beta^k \tilde A^k$, as in equation (6) of~\citet{bramoulle2009identification}, one obtains the reduced form 
\[Y = \frac{\alpha}{1 - \beta} + \gamma W + (\gamma \beta + \delta) \sum_{k=0}^\infty \beta^k \tilde A^{k+1} W + \sum_{k=0}^\infty \beta^k \tilde A^k \ep.\]
Unlike \citet{manski1993identification,bramoulle2009identification} and other works in the ``reflection problem'' literature, we are not concerned with the identification of the social effect parameters $\beta$ and $\delta$; these are only nuisance parameters toward the end of estimating $\tau$.  We do note, however, that conditions for identifiability are generally mild enough to be satisfied by real-world networks.  For example, \cite{bramoulle2009identification} show that the parameters in Model~\ref{model:lim-infinite} are identified whenever there exist a triple of individuals who are not all pairwise friends with each other; such a triple nearly certainly exists in any networks that we consider.

Now, let $X_{i,k}$ be the $i$-th coordinate of $\tilde A^k W$.  That is,
\begin{align*}
X_{i,1} &= \frac{1}{d_i} \sum_{j \in \cN_i} W_j \\
X_{i,2} &= \frac{1}{d_i} \sum_{j \in \cN_i} \frac{1}{d_j} \sum_{k \in \cN_j} W_k \\
X_{i,3} &= \frac{1}{d_i} \sum_{j \in \cN_i} \frac{1}{d_j} \sum_{k \in \cN_j} \frac{1}{d_k} \sum_{\ell \in \cN_k } W_\ell,
\end{align*}
and in general, for any $k \geq 1$, 
\[X_{i,k} = \frac{1}{d_i} \sum_{j_1 \in \cN_i} \frac{1}{d_{j_1}} \sum_{j_2 \in \cN_{j_1} }\dots \frac{1}{d_{j_k}} \sum_{j_{k-1} \in \cN_{j_{k-1}} } W_{j_k}.\]
Then Model~\ref{model:lim-infinite} is the same as
\begin{equation}
Y_i = \tilde \alpha + \tilde \gamma W_i +  \sum_{k=0}^\infty \tilde \beta_k X_{i,k} + \tilde \ep_i,
\label{eqn:lim-infinite}
\end{equation}
where we have reparametrized the coefficients as
\begin{align*}
\tilde \alpha &= \frac{\alpha}{1 - \beta} \\
\tilde \gamma &= \gamma \\
\tilde \beta_k &= (\gamma \beta + \delta) \beta^k \\
\tilde \ep &= \sum_{k=0}^\infty \beta^k A^k \ep.
\end{align*}
Notice that equation~\eqref{eqn:lim-infinite} respects exogeneity, as
\[\E[\tilde \ep | \+W] = \sum_{k=0}^\infty \beta^k A^k \E[\ep | \+W] = 0.\]
Each \hl{feature} $X_{i,k}$ represents the effect of treatments from units of graph distance $k$ on the response of unit $i$.  Since $|\beta| < 1$, the effects of the terms $\tilde \beta_k X_{i,k}$ do not contribute much to equation~\eqref{eqn:lim-infinite} when $k$ is large.  Therefore, for any finite integer $K$, we may consider approximating Model~\ref{model:lim-infinite} with a finite-dimensional model.
\begin{model}[Finite linear-in-means]
\label{model:lim-finite}
\begin{equation}
\label{eqn:lim-finite}
Y_i = \tilde \alpha + \tilde \gamma W_i + \sum_{k=0}^K \tilde \beta_k X_{i,k} + \tilde \ep_i,
\end{equation}
\end{model}
\end{subtheorem}
The approximation error is of order $\tilde \beta^{k+1} = (\gamma \beta + \delta)\beta^{K+1}$ (recall that $|\beta| < 1$).  Therefore, good estimates of the coefficients in equation~\eqref{eqn:lim-finite} should be good estimates of the coefficients in equation~\eqref{eqn:lim-infinite} as well.  Unless spillover effects are extremely large, the approximation may be quite good for even small values of $K$.  In fact, it may be reasonable to take equation~\eqref{eqn:lim-finite} rather than equation~\eqref{eqn:lim-infinite} as the truth where $K$ is no larger than the diameter of the network $G$, as spillovers for larger distances may not make sense.

As in Model~\ref{model:lim-exo}, we can consider the counterfactuals of interest.  If all units are globally exposed to treatment, then $W_i = 1$ and $X_{i,k} = 1$ for all $i$ and $k$.  Similarly, if all units are globally exposed to control, then $W_i = 0$ and $X_{i,k} = 0$ for all $i$ and $k$.  Therefore, by equation~\eqref{eqn:lim-infinite}, the estimand $\tau$ under Model~\ref{model:lim-infinite} is
\[\tau = \tilde \gamma + \sum_{k=0}^\infty \tilde \beta_k,\]
and under Model~\ref{model:lim-finite} it is
\[\tau = \tilde \gamma + \sum_{k=0}^K \tilde \beta_k.\]

Now, since Model~\ref{model:lim-finite} has only $K + 3$ coefficients, given $n > K + 3$ individuals one can estimate the coefficients using, say, ordinary least squares.  The treatment effect estimator
\[\hat \tau = \hat \gamma + \sum_{k=1}^K \hat \beta_k\]
is then unbiased for $\tau$ under Model~\ref{model:lim-finite} and ``approximately unbiased'' for $\tau$ under Model~\ref{model:lim-infinite}.  This discussion is of course quite informal, and we make more formal arguments in Section~\ref{sec:linear}.

One interpretation of the discussion in this section is that an endogeneous social effect in the linear-in-means model manifests as a propogation of exogenous effects through the social network, with the strength of the exogenous effect diminishing as the network distance increases.  Therefore, adjusting for the exogenous features within the first few neighborhoods is nearly equivalent to adjusting for the behavior implied by the endogenous social effect.


\subsection{Model assumptions and exposure models}
\hl{
The statements of the models discussed in this section couple together an interference mechanism restriction with a functional form assumption.  It is worth disentangling these assumptions and discussing why it may be sometimes advantageous for the analyst to consider them jointly.  First consider Model~\ref{model:lim-exo}.  It implies that the interference mechanism is restricted to influence from units only one step away in the graph, and furthermore, that this one-step influence is transmitted only through the statistic $X_i$.  This interference mechanism restriction can be framed in the language of \emph{constant treatment response} (CTR) mappings~\citep{manski2013identification}:
\begin{equation}
\mu_i(\+w) = \mu_i(\+w') \text{ for all } \+w, \+w' \in \cW \text{ such that } w_i = w_i', x_i = x_i'.
\label{eqn:ctr}
\end{equation}
The CTR statement~\eqref{eqn:ctr} is equivalent to specifying an exposure model that the potential outcomes depend only on $\+W$ through $W_i$ and $X_i$.  However, it makes no assumptions about the functional form of $\mu_i(\cdot)$, yet Model~\ref{model:lim-exo} goes further and makes a strong parametric functional form assumption about the response.  It is conceptually useful to recognize the different meanings and implications of these assumptions.

One tempting approach, then, might be for an analyst to first consider verifying whether the exposure model holds, using domain knowledge or otherwise.  The analyst then separately proceeds to consider appropriate functional forms (and perhaps only if nonparametric estimators exhibit low power).  This logic may succeed for simple exposures of the form implied by Model~\ref{model:lim-exo} but can lead to issues for more complex data-generating processes likely to be encountered in the real world.

This is made clear by the discussion of Models~\ref{model:lim-infinite} and~\ref{model:lim-finite}.  By rewriting Model~\ref{model:lim-infinite} as the infinite series given by equation~\eqref{eqn:lim-infinite}, we find that there is no data-reducing exposure model or CTR assumption that can handle such endogenous social effects!  This is discouraging unless the analyst jointly considers the parametric implications of an endogenous effect $|\beta| < 1$, which suggests a way forward via the finite approximation Model~\ref{model:lim-finite}.  Even if the linearity in equation~\eqref{eqn:lim-finite} is too strong, a natural relaxation might be a kind of generalized additive model of the form
\[Y_i = \alpha + \gamma W_i + \sum_{k=0}^K f(k) g(X_{i,k}) + \ep_i,\]
where $g$ is arbitrary but $f$ is restricted to be decreasing in $k$ in order to ensure that spillovers decrease in graph distance.

Furthermore, statisticians are well-versed in distinguishing and handling modeling violations of the mean function (here, corresponding to the interference function form) and the covariance function (corresponding to the interference restriction assumption), whereas statements like \eqref{eqn:ctr} may be a bit more abstruse for the practicing statistician.  The implications here are further clarified by the discussions in the following sections as well as the simulation examples provided in Section~\ref{sec:simulations}.
}

\section{Interference features and the general linear model}
\label{sec:linear}
In Section~\ref{sec:setup}, we showed that the mean function in the linear-in-means model is comprised of a linear combination of statistics $X_{i,k}$ which are constructed as functions of the treatment vector.  This fact suggests extending our approach to a linear model containing other functions of the treatment vector that are correlated with $Y_i$, not just the ones implied by the linear-in-means model.  We now formulate the general linear model.  We suppose that each unit $i$ is associated with a $p$-dimensional vector of \hl{\emph{interference features} or \emph{interference statistics}} $X_i \in \R^p$ that inform the pattern of interference for unit $i$.  We assume that the $X_i$ are low-dimensional ($p \ll n$).
Because $X_i$ is to be used for adjustment, the main requirement is that it not be a ``post-treatment variable''; that is, that it not be correlated with the treatment $W_i$.  Therefore, we require the following assumption:
\begin{assumption}
\label{asm:x-indep-w}
$X_i \indep W_i$ for all $i \in [n]$.
\end{assumption}
Let $\Wni$ denote the vector of \emph{indirect treatments}, which is the $n-1$ vector of all treatments except for $W_i$.  The key feature of our approach is that even though $X_i$ must be independent of $W_i$, \emph {it is not necessary that $X_i$ be independent of the vector of indirect treatments $\Wni$.}  In fact, in order for $X_i$ to be useful for adjusting for interference, we expect that $X_i$ will be correlated with some entries of $\Wni$.  In particular, $X_i$ may be a deterministic function $x_i(\cdot)$ of the indirect treatments,
\begin{equation}
\label{eqn:x-determ-fn}
X_i = x_i(\Wni).
\end{equation}
Adjusting for such a variable $X_i$ will not cause post-treatment adjustment bias as long as the entries of $\+W$ are independent of each other.  This holds automatically in a Bernoulli randomized design (Assumption~\ref{asm:bernoulli-design}).  

The \hl{features} $X_i$ may depend on static structural information about the units such as network information provided by $G$, though since $G$ is static we supress this dependence in the notation.  For example, $X_i$ defined as in equation~\eqref{eqn:frac-nbh}, which represents the proportion of treated neighbors, captures a particular form of exogenous social influence.  Provided there are no self-loops in $G$ so that $A_{ii} = 0$, $W_i$ does not appear on the right-hand side of equation~\eqref{eqn:frac-nbh} and so $X_i$ and $W_i$ are independent.

We assume that we can easily sample from the distribution of $X_i$.  In particular, if $X_i = x_i(\Wni)$, then the distribution of $X_i$ can be constructed by Monte Carlo sampling from the randomization distribution of the treatment $\+W$.  In this paper and in all the examples we use, we assume that $X_i$ is a function of $\Wni$ as in equation~\eqref{eqn:x-determ-fn}, so that conditioning on $\Wni$ removes all randomness in $X_i$.  But the generalization is easily handled.

In this section we assume that the response is linear in $X_i$; we address nonparametric response surfaces in Section~\ref{sec:nonparametric}.
\begin{model}[Linear model]
\label{model:linear}
Given $X_i$, let the response $Y_i$ follow
\[Y_i = W_i \mu^{(1)}(X_i) + (1 - W_i)\mu^{(0)}(X_i) + \ep_i,\]
where the conditional response surfaces
\[\mu^{(0)}(x) = \E[Y_i^{(0)} | X = x], \qquad \mu^{(1)}(x) = \E[Y_i^{(1)} | X = x]\]  satisfy
\[\mu^{(0)}(x) = \beta_0^\top x, \qquad \mu^{(1)}(x) = \beta_1^\top  x\]
for $x \in \R^p$ and $\beta_0, \beta_1 \in \R^p$. 
That is, they follow a ``separate slopes'' linear model in $X_i$.  We assume $p < n$. 
\end{model}

In the above parametrization, we assume that the first coordinate of each $X_i$ is set to $1$, so that the vectors $\beta_0$ and $\beta_1$ contain coefficients corresponding to the intercept as in the classical OLS formulation.

\subsection{Feature engineering}
Before considering assumptions on the residuals $\ep_i$, we pause here to emphasize the flexibility provided by modeling the interference pattern as in Model~\ref{model:linear}.  In this framework, the researcher can use domain knowledge to construct graph features that are expected to contribute to interference.  In essence, we have transformed the problem of determining the structure of the interference pattern into a feature engineering problem, which is perhaps a more intuitive and accessible task for the practitioner.

To elaborate, consider the problem of selecting an exposure model.  \citet{ugander2013graph} propose and study a number of different exposure models for targeting the global treatment effect, including \emph{fractional exposure} (based on the fraction of treated neighbors), \emph{absolute exposure} (based on the raw number of treated neighbors), and extensions based on the $k$-core structure of the network.  In reality, it may be the case that fractional and absolute exposure both contribute partial effects of interference, so ideally one wishes to avoid having to choose between one of the two exposure models.  On the other hand, both features are easily included in Model~\ref{model:linear} by encoding both the fraction and raw number of treated neighbors in $X_i$.  \hl{(Including both features only makes sense when working with a complex network.  If the interference structure is comprised of large, disjoint, and equally-sized clusters, as in partial interference, then the fraction and number of treated neighbors encode roughly the same information and one obtains a collinearity scenario that violates the full-rank assumption of Proposition~\ref{prop:linear-unbiased}.  The methods in this paper are primarily motivated by the complex network setting.)}

In a similar manner, the researcher may wish to handle longer-range interference, such as that coming from two-step or greater neighborhoods.  It is possible to handle two-step information by working with the graph corresponding to the adjacency matrix $A^2$, but this approach is unsatisfactory because presumably one-step interference is stronger than two-step interference, and this distinction is lost by using $A^2$.  On the other hand, if one-step and two-step network information are encoded as separate features, both effects are included and the magnitudes of their coefficients will reflect the strength of the corresponding interference contributed by each feature.

Furthermore, nothing in our framework requires the variables to be constructed from a single network.  Often, the researcher has access to multiple networks defined on the same vertex set---i.e., a \emph{multilayer network}~\citep{kivela2014multilayer}---representing different types of interactions among the units.  For example, social networking sites such as Facebook and Twitter contain multiple friendship or follower networks based on the strength and type of interpersonal relationship (e.g.\ family, colleagues, and acquaintances), as well as activity-based networks constructed from event data such as posts, tweets, likes, or comments.  Often these networks are also dynamic in time.  Given the sociological phenomenon that the strength of a tie is an indicator of its capacity for social influence~\citep{granovetter1973strength} and that people use different mediums differently when communicating online~\citep{haythornthwaite1998work}, any or all of these network layers can conceivably be a medium for interference in varying amounts depending on the treatment variable and outcome metric in question.  In our framework graph features from different network layers are easily included in the model.

\subsection{Exogeneity assumptions}

Consider the following assumptions on the residuals.
\begin{assumption}
\label{asm:errors}
\begin{enumerate}[(a)]
\item  The errors are strictly exogenous: $\E[\ep_i | X_1, \dots, X_n] = 0$ for all $i \in [n]$. \label{asm:errors-exogenous}
\item The errors are independent. \label{asm:errors-independent}
\item The errors are homoscedastic: $\var(\ep_i | X_1, \dots, X_n) = \sigma^2$ for all $i \in [n]$.  \label{asm:errors-homoscedastic}
\end{enumerate}
\end{assumption}


\hl{Assumption~\ref{asm:errors}(\ref{asm:errors-exogenous}) captures the requirement that the features contain all of the information needed to adjust for the bias contributed by interference, and thus is similar to an unconfoundedness or ignorability assumption often invoked in observational studies. Point estimates can be constructed based only on Assumption 3(a), but variance estimation requires Assumption~\ref{asm:errors}(\ref{asm:errors-independent}) so that each data point contributes additional independent information.  Note that in the SUTVA case Assumption~\ref{asm:errors}(\ref{asm:errors-exogenous}) is all that is needed for valid inference.  However under interference, it is possible that conditioning on the features removes all bias but interference is still present in the errors, in which case i.i.d.-based standard errors would be incorrect.}  These assumptions cannot be verified from the data, and so this setup borrows all of the problems that come with selecting an exposure model or being able to verify unconfoundedness. However, our setup is slightly different because of the flexibility afforded by the \hl{features}. Compared to what we envision as the usual observational studies setting, our \hl{features} are constructed from the treatment vector and social network rather than being collected in the wild, and so they are quite cheap to construct via feature engineering. That said, more work for conducting sensitivity analysis for interference or spillover effects is certainly needed.

Assumption~\ref{asm:errors}(\ref{asm:errors-homoscedastic}) is the easiest to deal with if violated. One may use a heteroscedisticity-consistent estimate of the covariance matrix, also known as the sandwich estimator or the Eicker-Huber-White estimator~\citep{eicker1967limit,huber1967behavior,white1980heteroskedasticity}.  In this paper we invoke Assumption~\ref{asm:errors}(\ref{asm:errors-homoscedastic}) mainly to simplify notation, but heteroscedasticity-robust extensions are straightforward.

For $X_i$ following equation~\eqref{eqn:x-determ-fn}, denote
\[X_i^{(\+0)} = x_i(\Wni = \+0), \qquad X_i^{(\+1)} = x_i(\Wni = \+1).\]
The variable $X_i^{(\+0)}$ represents the context for unit $i$ under the counterfactual scenario that $i$ is exposed to global control, and the variable $X_i^{(\+1)}$ represents the context for unit $i$ under the counterfactual scenario that $i$ exposed to global treatment.  Both of these values are non-deterministic.\footnote{
In the event $X_i$ are not defined through a function $x_i(\cdot)$, one may work with $X_i^{(\+0)} = X_i | (\Wni = \+0)$ and $X_i^{(\+1)} = X_i | (\Wni = \+1)$,
where this notation means that $X_i^{(\+0)}$ follows the conditional distribution of $X_i$, conditionally on the event that $\Wni = \+0$, and similarly for $X_i^{(\+1)}$.  In this case $X_i^{(\+0)}$ and $X_i^{(\+1)}$ may be random, and estimands can be defined using $\E[X_i^{(\+0)}]$ and $\E[X_i^{(\+1)}]$ instead.}
For example, if $X_i$ be the ``mean treated'' \hl{statistic} as defined in equation~\eqref{eqn:frac-nbh}, then $X_i^{(\+0)} = 0$ and $X_i^{(\+1)} = 1$ for every unit $i \in [n]$.

We now consider the estimand under Model~\ref{model:linear} and Assumption~\ref{asm:errors}.  The GATE for Model~\ref{model:linear} is
\begin{align*}
\tau &= \avgin \left[\E[Y_i | \+W = \+1] - \E[Y_i | \+W = \+0]\right] \\
&= \avgin \left[\mu^{(1)}(X_i^{(\+1)}) - \mu^{(0)}(X_i^{(\+0)})\right] \\
&= \avgin \left[(X_i^{(\+1)})^\top \beta_1 - (X_i^{(\+0)})^\top \beta_0 \right],
\end{align*}
where the second equality is by Assumption~\ref{asm:errors}(\ref{asm:errors-exogenous}).  Now introduce the quantities
\[\omega_0 = \avgin X_i^{(\+0)}, \qquad \omega_1 = \avgin X_i^{(\+1)},\]
which are the mean counterfactual \hl{feature} values for global control and global treatment, averaged over the population.  We emphasize that $\omega_0$ and $\omega_1$ are non-deterministic and known, because the distribution of $X_i$ is assumed to be known.
We then have
\begin{equation}
\label{eqn:linear-tau}
\tau = \omega_1^\top \beta_1 - \omega_0^\top\beta_0.
\end{equation}
Such an estimand, which focuses on the \hl{statistics} of the finite population at hand, is natural in the network setting where there is no clear superpopulation or larger network of interest.  

We now construct an estimator by estimating the regression coefficients with ordinary least squares.  For $w = 0, 1$, let $X_w$ be the $N_w \times p$ design matrix corresponding to \hl{features} belonging to treatment group $w$, where the first column of $X_w$ is a column of ones.  Let $y_w$ be the $N_w$-vector of observed responses $Y_i$ for treatment group $w$. 
Then we use the standard OLS estimator
\begin{equation}
\label{eqn:beta-hat}
\hat \beta_w = (X_w^\top X_w)^{-1} X_w^\top y_w.
\end{equation}
The estimate of the treatment effect is taken to be the difference in mean predicted outcomes under the global treatment and control counterfactual distributions,
\begin{align}
\hat \tau &= \omega_1^\top\hat \beta_1 - \omega_0^\top\hat \beta_0. \label{eqn:linear-tau-hat} 
\end{align}
Assuming Model~\ref{model:linear} holds, $\hat \tau$ is an unbiased estimate of $\tau$, which follows from unbiasedness of the OLS coefficients.
\begin{proposition}
\label{prop:linear-unbiased}
Suppose Model~\ref{model:linear} and Assumptions~\ref{asm:bernoulli-design}, \ref{asm:x-indep-w}, and~\ref{asm:errors}(\ref{asm:errors-exogenous}) hold.  Let $\tau$ and $\hat \tau$ be defined as in equations~\eqref{eqn:linear-tau} and~\eqref{eqn:linear-tau-hat}, and let $\hat \beta_w$ for $w = 0, 1$ be OLS estimators as defined in equation~\eqref{eqn:beta-hat}.  Then conditionally on $X_w$ being full (column) rank,\footnote{
	Since $X_w$ is random and depends on $\+W$, conditioning on $X_w$ having full column rank is necessary, even though this condition may not be fulfilled for all realizations of the treatment vector.  For example, if $X_w$ contains a column for the fraction of neighbors treated, then it is possible though highly unlikely for all units to be assigned to treatment, in which case this column is collinear with the intercept and $X_w$ is not full rank.  We shall, for the most part, ignore this technicality and assume that the \hl{features} are chosen so that the event that $X_w^\top X_w$ is singular doesn't happen very often, and is in fact negligible asymptotically.  \hl{Understanding  combinations of network structures and interference mechanisms that give rise to singular $X_w$ is of interest to practitioners but outside the scope of our study here.}
} 
$\hat \beta_w$ is an unbiased estimator of $\beta_w$ and $\hat \tau$ is an unbiased estimator of $\tau$.  
\end{proposition}

(Proofs for Proposition~\ref{prop:linear-unbiased} and other results are deferred to the appendix.) Notice that the treatment group predicted mean is
\[\omega_1^\top \hat \beta_1 = \omega_1^\top(X_1^\top X_1)^{-1} X_1^\top y_1\]
and the control group predicted mean is
\[\omega_0^\top \hat \beta_0 = \omega_0^\top(X_0^\top X_0)^{-1} X_0^\top y_0.\]
Therefore $\hat \tau$ is linear in the observed response vector $y$.  
That is, $\tau = a_0^\top y_0 + a_1^\top y_1$ where the weight vectors $a_0 \in \R^{N_0}$ and $a_1 \in \R^{N_1}$ are given by
\begin{align}
a_0^\top &= \omega_0^\top(X_0^\top X_0)^{-1} X_0^\top \label{eqn:ols-weight0} \\
a_1^\top &= \omega_1^\top(X_1^\top X_1)^{-1} X_1^\top. \label{eqn:ols-weight1}
\end{align}
These weights allow us to compare the reweighting strategy with that of other linear estimators, such as the H\'ajek estimator, which is a particular weighted mean of $y$.  More details are provided in Section~\ref{sec:sim-hajek-description}, with an example given in Section~\ref{sec:sim-weights}.

\subsection{Inference}

Now we provide variance expressions under the assumption that the errors are exogenous, independent, and homoscedastic, as in Assumption~\ref{asm:errors}.

\begin{theorem}
\label{thm:ols-variance-finite}
Suppose Model~\ref{model:linear} and Assumptions~\ref{asm:bernoulli-design}, \ref{asm:x-indep-w}, and~\ref{asm:errors} hold.  Then
\begin{equation}
\label{eqn:ols-variance-finite}
\var(\hat \tau) = \sigma^2 (\|\omega_0\|_{\Gamma_0}^2 + \|\omega_1\|_{\Gamma_1}^2),
\end{equation}
where $\|v\|_M^2 = v^\top M v$, and $\Gamma_w = \E[(X_w^\top X_w)^{-1}]$, and $\omega_w$ is the mean of the counterfactual \hl{feature} distribution (including an intercept) for $w = 0, 1$.
\end{theorem}

\subsubsection{Variance estimation}
In order to estimate the variance~\eqref{eqn:ols-variance-finite}, we must estimate the quantities $\Gamma_0 = \E[(X_0^\top X_0)^{-1}]$ and $\Gamma_1 = \E[(X_1^\top X_1)^{-1}]$, which are the expected inverse sample covariance matrices.  Of course, $(X_0^\top X_0)^{-1}$ and $(X_1^\top X_1)^{-1}$ are observed and unbiased estimators.    However, unlike standard \hl{baseline characteristics} collected in the wild, we envision that the $X_i$ are constructed from the graph $G$ and the treatment vector $\+W$, and so we can take advantage of the fact that the distribution of $X_i$ is completely known to the researcher.  It is thus possible to compute $\Gamma_0$ and $\Gamma_1$ up to arbitrary precision by repeated Monte Carlo sampling from the randomization distribution of $\+W$.   For clarity, this estimation procedure is illustrated in Algorithm~\ref{alg:mc-var-ols}.

\begin{algorithm}[t]
\caption{Estimating $\Gamma_0$ and $\Gamma_1$ by Monte Carlo}
\label{alg:mc-var-ols}
\begin{algorithmic}
\FOR{b = 1:B} 
  \STATE{Sample treatment $\+W_b \in \cW$ and compute corresponding \hl{features} $X_{b,i}$ and sample sizes $N_{b,0}$ and $N_{b,1}$}
  \STATE{Calculate sample covariances
    \begin{align*}
    (\tilde X_0^\top \tilde X_0)_b &\leftarrow \frac{1}{N_{b,0}} \sumin (1 - W_{b,i})X_{b,i}X_{b,i}^\top \\
    (\tilde X_1^\top \tilde X_1)_b &\leftarrow \frac{1}{N_{b,1}} \sumin W_{b,i}X_{b,i}X_{b,i}^\top
    \end{align*}
  } 
\ENDFOR
\RETURN{Moment estimates
  \begin{align*}
  \hat \Gamma_w &\leftarrow \widehat \E[(\tilde X_w^\top \tilde X_w)^{-1}] = \frac{1}{B} \sum_{b=1}^B (\tilde X_w^\top \tilde X_w)_b^{-1}
  \end{align*}
  for $w = 0, 1$.
}
\end{algorithmic}
\end{algorithm}

Finally, we can estimate $\sigma^2$ in the usual way, with the residual mean squared error
\[\hat \sigma^2 = \avgin \left(Y_i - W_i(\hat \beta_1^\top  X_i) - (1 - W_i)(\hat \beta_0^\top X_i)\right)^2.\]
Equipped with $\hat \sigma^2$ and Monte Carlo estimates $\hat \Gamma_w$, we can use the variance estimate
\begin{equation}
\label{eqn:ols-varest}
\widehat{\var}(\hat \tau) = \hat \sigma^2\left(\|\omega_0\|_{\hat \Gamma_0}^2 + \|\omega_1\|_{\hat \Gamma_1}^2\right).
\end{equation}


\subsection{Asymptotic results}

Proposition~\ref{prop:linear-unbiased} and Theorem 1 characterize the finite $n$ expectation and variance of the treatment effect estimator under Model 4.  Establishing an asymptotic result is more nuanced, as because of the dependence among units implied by interference, the quantities $\E[(X_0^\top X_0)^{-1}]$ and $\E[(X_1^\top X_1)^{-1}]$ may not be $O(n^{-1})$ in which case $\hat \tau$ would not converge at a $\sqrt{n}$ rate.  This is a problem with dealing with interference in general, making comparisons to the semiparametric efficiency bound~\citep{hahn1998role}, a standard benchmark in the SUTVA case, difficult in this setting.  However we can state a $\sqrt{n}$ central limit theorem in the event that the sample mean and covariance do scale and converge appropriately.  To do so, we implicitly assume existence of a sequence of populations indexed by their size $n$, and that the parameters associated with each population setup, such as $\beta_0$, $\beta_1$, $\pi$, and $\sigma^2$, converge to appropriate limits.  Such an asymptotic regime is the standard for results of this sort~\citep[cf.][]{freedman2008regressionA,freedman2008regressionB,lin2013agnostic,abadie2017should,abadie2017sampling,savje2017average,chin2018central}.  We suppress the index on $n$ to avoid notational clutter.

So that we can compare to previous works, it is helpful to reparametrize the linear regression setup so that the intercept and slope coefficients are written separately.  That is, let $X_i$ and $\omega_w$ be redefined to exclude the intercept, and let $\beta_w = (\alpha_w, \eta_w)$ so that the mean functions are written $\mu^{(w)}(x) = \alpha_w + \eta_w^\top x$, where $\alpha_w$ is the intercept parameter and $\eta_w$ is the vector of slope coefficients.  Then the GATE is
\[\tau = (\alpha_1 + \omega_1^\top \eta_1) - (\alpha_0 + \omega_0^\top \eta_0).\]
Denote the within-group sample averages by
\[\bar y_1 = \frac{1}{N_1} \sumin W_i Y_i, \qquad \bar y_0 = \frac{1}{N_0} \sumin (1 - W_i) Y_i\]
and
\[\bar X_0 = \frac{1}{N_1} \sumin W_i X_i, \qquad \bar X_0 = \frac{1}{N_0} \sumin (1 - W_i) X_i.\]  
Since the intercept is determined by $\hat \alpha_w = \bar y_w - \bar X_w^\top \hat \eta_w$,
the estimator $\hat \tau$, equation~\eqref{eqn:linear-tau-hat}, is written as
\begin{align}
\hat \tau &= (\hat \alpha_1 + \omega_1^\top \hat \eta_1) - (\hat \alpha_0 + \omega_0^\top \hat \eta_0). \nonumber\\
&= \bar y_1 - \bar y_0 + (\omega_1 - \bar X_1)^\top \hat \eta_1 - (\omega_0 - \bar X_0)^\top \hat \eta_0. \label{eqn:ols-estimator-with-int}
\end{align}
Now, $\hat \tau$ is seen to be an adjustment of the difference-in-means estimator $\bar y_1 - \bar y_0$.  The adjustment depends on both the estimated strength of interference, $\hat \eta_w$, and the discrepancy between the means of the observed distribution and the reference or target distribution, $\bar X_w - \omega_w$.  This linear shift is a motif in the regression adjustment literature, and is reminiscent of, e.g., equation (16) of~\citet{aronow2013class}.

We now state a central limit theorem for $\hat \tau$.  
\begin{theorem}
\label{thm:ols-clt}
Assume the setup of Theorem~\ref{thm:ols-variance-finite}.  Assume further that the sample moments converge in probability:
\begin{align*}
\bar X = \avgin X_i& \pto \mu_X, \\
S = \avgin (X_i - \bar X)^\top (X_i - \bar X) &\pto \Sigma_X,
\end{align*}
where $\Sigma_X$ is positive definite, and that all fourth moments are bounded.
Then $\sqrt{n}(\hat \tau - \tau) \wto N(0, V)$, 
where 
\begin{equation}
\label{eqn:asymp-var}
V = \sigma^2 \left(\frac{1}{\pi(1 - \pi)} + \frac{\|\omega_0 - \mu_X\|_{\Sigma_X^{-1}}^2}{1 - \pi} + \frac{\|\omega_1 - \mu_X\|_{\Sigma_X^{-1}}^2}{\pi}\right).
\end{equation}
\end{theorem}
The terms in expression~\eqref{eqn:asymp-var} are unpacked versions of the terms in expression~\eqref{eqn:ols-variance-finite}, and can be stated in this way since the \hl{feature} moments converge at the appropriate rate.


\subsection{Relationship with standard regression adjustments}
The practitioner may also wish to perform standard regression adjustments to adjust for static, contextual node-level variables such as age, gender, and other demographic variables.  This fits easily into the framework of Model 4, as any such static variable $X_i$ can be viewed as simply a constant function of the indirect treatment vector $\Wni$.  Then the adjustment is not used to remove bias but simply to reduce variance by balancing the \hl{feature} distributions.  In this case the counterfactual (global exposure) distribution is the same as the observed distribution, and in particular, $\omega_0 = \mu_X$ and $\omega_1 = \mu_X$.  Hence we see that Theorem~\ref{thm:ols-clt} reduces to the standard asymptotic result for regression adjustments using OLS.
\begin{corollary}
\label{cor:ols-fixed}
Assume the setup of Theorem~\ref{thm:ols-clt}.  Suppose $X_i$ is independent of $\Wni$.  Then $\omega_0 = \mu_X$ and $\omega_1 = \mu_X$ and 
\[\sqrt{n}(\hat \tau - \tau) \wto N\left(0, \frac{\sigma^2}{\pi(1 - \pi)}\right),\]
\end{corollary}
This variance in Corollary~\ref{cor:ols-fixed} is the same asymptotic variance as in the standard regression adjustment setup~\citep[cf.][Theorem 2]{wager2016high}.  In practice, if some components of $X_i$ are static covariates and some are interference variables, then the resulting variance will be decomposed into the components stated in Theorem~\ref{thm:ols-clt} and Corollary~\ref{cor:ols-fixed}.  \hl{Conditioning on both baseline covariates and interference features may in fact be necessary to ensure that Assumption~\ref{asm:errors}(\ref{asm:errors-exogenous}) holds.  For example if $X_i$ is the number of treated neighbors it may be believed that the potential outcomes depend on node degree (in the graph $G$) as well.}

\section{Nonparametric adjustments}
\label{sec:nonparametric}
In this section we relax the linear model, Model~\ref{model:linear}:
\begin{model}[Non-linear response surface]
\label{model:nonparametric}
Let $Y_i$ follow
\[Y_i = W_i \mu^{(1)}(X_i) + (1 - W_i)\mu^{(0)}(X_i) + \ep_i,\]
with conditional mean response surfaces
\[\mu^{(0)}(x) = \E[Y_i^{(0)} | X = x], \qquad \mu^{(1)}(x) = \E[Y_i^{(1)} | X = x].\]
We make no parametric assumptions on the form of $\mu^{(0)}(x)$ and $\mu^{(1)}(x)$.
\end{model}
We maintain Assumption~\ref{asm:errors}, namely that SUTVA holds conditionally on $X_1, \dots, X_n$.  

In the SUTVA setting, adjustment with OLS works best when the adjustment variables are highly correlated with the potential outcomes; that is, the precision improvement largely depends on the prediction accuracy.  This fact suggests that predicted outcomes obtained from an arbitrary machine learning model can be used for adjustment, an idea formalized by~\citet{wager2016high,wu2017loop}.  Based on ideas from~\citet{aronow2013class}, these papers propose using the estimator
\begin{equation}
\label{eqn:crossfit-sutva}
\frac{1}{n} \sumin \left(\hat \mu_{-i}^{(1)}(X_i) - \hat \mu_{-i}^{(0)}(X_i)\right) + \frac{1}{N_1} \sumin W_i\left(Y_i - \hat \mu_{-i}^{(1)}(X_i)\right) - \frac{1}{N_0} \sumin (1 - W_i)\left(Y_i - \hat \mu_{-i}^{(0)}(X_i)\right),
\end{equation}
where $\hat \mu_{-i}^{(0)}$ and $\hat \mu_{-i}^{(1)}$ are predictions of the potential outcomes obtained without using the $i$-th observation.  
This doubly-robust style approach is called \emph{cross-estimation} by~\citet{wager2016high} and the \emph{leave-one-out potential outcomes} (LOOP) estimator by~\citet{wu2017loop} who focus on imputing the outcomes using a version of leave-one-out cross validation. This estimator is also reminiscent of the \emph{double machine learning} (DML) cross-fitting estimators developed for the observational study setting~\citep{chernozhukov2018double}, which consists of the following two-stage procedure: (a) train predictive machine learning models $\hat e(\cdot)$ of $X_i$ on $W_i$ (the propensity model) and $\hat m(\cdot)$ of $X_i$ on $Y_i$ (the response model), and then (b) use the out-of-sample residuals $W_i - \hat e(X_i)$ and $Y_i - \hat m(X_i)$ in a final stage regression.  The difference in the experimental setting is that the propensity scores are known and so no propensity model is needed.  \citet{wu2017loop} study the behavior of~\eqref{eqn:crossfit-sutva} in the finite population setting where the only randomization comes from the treatment assignment, and~\citet{wager2016high} provide asymptotic results for estimating the population average treatment effect.  As long as the predicted value $\hat \mu_{-i}^{(w)}$ does not use the $i$-th observation, estimator~\eqref{eqn:crossfit-sutva} allows us to obtain asymptotically unbiased adjustments and valid inference using machine learning algorithms such as random forests or neural networks.  In practice, such predictions are obtained by a cross validation-style procedure in which the data are split into $K$ folds, and the predictions for each fold $k$ are obtained using a model fitted on data from the other $K-1$ folds.  (Cross validation on graphs is in general difficult~\citep{chen2018network,li2018network}, but our procedure is unrelated to that problem because the features are constructed from the entire graph and fixed beforehand.) 

In this section we apply insights from the above works to the interference setting.  Under Model~\ref{model:nonparametric}, the global average treatment effect has the form
\[\tau = \avgin \left[\mu^{(1)}(X_i^{(\+1)}) - \mu^{(0)}(X_i^{(\+0)})\right].\]
To develop an estimator of $\tau$, consider the form of the OLS estimator given by equation~\eqref{eqn:ols-estimator-with-int}, which can be rewritten as
\begin{align}
\hat \tau &= \bar y_1 - \bar y_0 + (\omega_1 - \bar X_1)^\top \hat \eta_1 - (\omega_0 - \bar X_0)^\top \hat \eta_0 \nonumber\\
&= \omega_1^\top \hat \eta_1 - \omega_0^\top \hat \eta_0 + (\bar y_1 - \bar X_1^\top \hat \eta_1) - (\bar y_0 - \bar X_0^\top \hat \eta_0) \nonumber\\
&= \avgin \left((X_i^{(\+1)})^\top  \hat \eta_1 - (X_i^{(\+0)})^\top \hat \eta_0\right) + \frac{1}{N_1} \sumin W_i \left(Y_i - X_i^\top \hat \eta_1\right) - \frac{1}{N_0} \sumin (1 - W_i) \left(Y_i - X_i^\top \hat \eta_0\right). \label{eqn:ols-crossfit-motivator}
\end{align}
Now, by analog, we define the estimator for the nonparametric setting as
\begin{equation}
\label{eqn:tauhat-nonlinear}
\hat \tau = \avgin \left(\hat \mu_{-i}^{(1)}(X_i^{(\+1)}) - \hat \mu_{-i}^{(0)}(X_i^{(\+0)})\right) + \frac{1}{N_1} \sumin W_i\left(Y_i - \hat \mu_{-i}^{(1)}(X_i)\right) - \frac{1}{N_0} \sumin (1 - W_i)\left(Y_i - \hat \mu_{-i}^{(0)}(X_i)\right).
\end{equation}
One sees that equations~\eqref{eqn:ols-crossfit-motivator} and~\eqref{eqn:tauhat-nonlinear} agree whenever $\hat \mu^{(w)}(x) = \hat \alpha_w + x^\top \hat \eta_w$.  Furthermore, equation~\eqref{eqn:tauhat-nonlinear} is equal to its SUTVA version, equation~\eqref{eqn:crossfit-sutva}, whenever $X_i^{(\+0)} = X_i^{(\+1)} = X_i$.

Because the units can be arbitrarily connected, the cross-fitting component partitions are not immediately guaranteed to be exactly independent, and so any theoretical guarantees must assume some form of approximate independence of the out-of-sample predictions.  In this work we leave such theoretical results open for future work; our primary contribution is the proposal of estimator~\eqref{eqn:tauhat-nonlinear} and a bootstrap variance estimation method that respects the empirical structure of interference.

\subsection{Bootstrap variance estimation}
\label{sec:nonparametric-varest}
Here we discuss a method for placing error bars on the estimate $\hat \tau$ defined in equation~\eqref{eqn:tauhat-nonlinear}.  
We propose using a bootstrap estimator to estimate the sampling variance.  Under exogeneity (Assumption~\ref{asm:errors}), the \hl{features} and residuals contribute orthogonally to the total variance, and so the model and residuals can be resampled separately.  

Instead of using the fixed, observed $X_1, \dots X_n$ as in a standard residual bootstrap, we propose capturing the entire variance induced by the \hl{feature} distribution by sampling a new $X_i$ from its population distribution for each bootstrap replicate.  That is, for each of $B$ bootstrap repetitions, we sample a new treatment vector $\+W^b$ and compute bootstrapped \hl{features} $X_i^b = x_i(\Wni^b)$.  The means are then computed using the fitted function as $\hat \mu_{-i}^{(0)}(X_i^b)$ and $\hat \mu_{-i}^{(1)}(X_i^b)$.  Provided that the adjustments are consistent in the sup norm sense, that is, that
\[\sup_x |\hat \mu^{(0)}(x) - \mu^{(0)}(x)| \pto 0, \qquad \sup_x |\hat \mu^{(1)}(x) - \mu^{(1)}(x)| \pto 0,\]
then $\hat \mu^{(0)}(\cdot)$, $\hat \mu^{(1)}(\cdot)$ serve as appropriate stand-ins for $\mu^{(0)}(\cdot)$, $\mu^{(1)}(\cdot)$ in large samples.

For the residual portion, we take the initial fitting functions $\hat \mu_{-i}^{(0)}(\cdot)$ and $\hat \mu_{-i}^{(1)}(\cdot)$ and compute the residuals
\[\hat \ep_i = Y_i - W_i \hat \mu_{-i}^{(1)}(X_i) - (1 - W_i) \hat \mu_{-i}^{(0)}(X_i).\] 
Under an assumption of independent errors, it is appropriate to compute bootstrap residuals $\ep_1^b, \dots, \ep_n^b$ by sampling with replacement from the observed residuals $\hat \ep_1, \dots, \hat \ep_n$.
We can then construct an artificial bootstrap response
\[Y_i^b = W_i^b \hat \mu_{-i}^{(1)}(X_i^b) + (1 - W_i^b) \hat \mu_{-i}^{(0)}(X_i^b) + \ep_i^b.\]
We then compute $\hat \tau^b$ using data $(Y_i^b, X_i^b, W_i^b)$, and then take the bootstrap distribution $\{\hat \tau^b\}_{b=1}^B$ as an approximation to the true distribution of $\hat \tau$.  To construct a $1 - \alpha$ confidence interval, one can calculate the endpoints using approximate Gaussian quantiles,
\[\hat \tau \pm z_{\alpha / 2} \sqrt{\var(\hat \tau_b)}.\]
Alternatively, one may use the
$\alpha / 2$ and $1 - \alpha / 2$ quantiles of the empirical bootstrap distribution (a percentile bootstrap), which is preferable if the distribution of $\hat \tau$ is skewed.


\hl{We wish to emphasize that the main insight here is that exogeneity allows the feature and residual variances to be handled separately, and that the feature variance can be computed from the design, however complicated the structure of $X_i$ itself may be.  The bootstrap residuals $\ep_i^b$ as described above rely on independent errors, but in fact the practitioner is free to utilize the entirety of the rich bootstrap literature stemming from~\citet{efron1979bootstrap} in the event that this independence assumption is violated.  For example, one may use versions of the block bootstrap~\citep{kunsch1989jackknife} to try and protect against correlated errors.  One can use more complicated bootstrap methods to be more faithful to the empirical distribution, such as incorporating higher-order features of the distribution via bias-corrected and accelerated (BCa) intervals~\citep{efron1987better}, or handling heteroscedasticity via the wild bootstrap~\citep{wu1986jackknife}.}

\section{Simulations}
\label{sec:simulations}
This section is devoted to running a number of simulation experiments.  Our goals in  these simulations are to (a) verify that our adjustment estimators and variance estimates are behaving as intended, (b) compare the performance of our proposed estimators to that of existing inverse propensity weighted estimators based on exposure models, and (c) empirically explore the behavior of our estimators in regimes of mild model misspecification.

\subsection{Simulation setup and review of exposure modeling}
\label{sec:sim-hajek-description}

For the network $G$ we use a subset of empirical social networks from the \texttt{facebook100} dataset, an assortment of complete online friendship networks for one hundred colleges and universities collected from a single-day snapshot of Facebook in September 2005.  A detailed analysis of the social structure of these networks was given in~\citet{traud2011comparing,traud2012social}. 
We use an empirical network rather than an instance of a random graph model in order to replicate as closely as possible the structural characteristics observed in real-world networks.  We use the largest connected components of the Caltech and Stanford networks.  Some summary statistics for the networks are given in Table~\ref{table:fb-summary}.

\begin{table}[ht]
\centering
\begin{tabular}{r|rr}
\toprule 
network & Caltech & Stanford \\
\midrule
number of nodes & 762 & 11586 \\
number of edges & 16651 & 568309 \\
diameter & 6 & 9 \\
average pairwise distance & 2.33 & 2.82 \\
\bottomrule
\end{tabular}
\caption{Summary statistics for the \texttt{facebook100} networks.}
\label{table:fb-summary}
\end{table}

In all simulation regimes we compare our regression estimators to two other estimators, which we describe now.  As a baseline we use the SUTVA difference-in-means estimator
\[\DM = \frac{1}{N_1} \sumin W_i Y_i - \frac{1}{N_0} \sumin (1 - W_i)Y_i.\]

\subsubsection{Exposure modeling IPW estimators}
We also compare to an inverse propensity weighted estimator derived from a local neighborhood exposure model.  We now briefly describe the exposure model-based estimators framed in the language of \emph{constant treatment response} assumptions~\citep{manski2013identification}.   For $Y_i(\+w) = \mu_i(\+w) + \ep_i$, this approach partitions the space of treatments $\cW$ into classes of treatments that map to the same mean response $\mu_i(\cdot)$ for unit $i$.  The partition function is assumed known, and is called an \emph{exposure function}.  The no-interference portion of SUTVA can be specified as an exposure model, since no-interference is equivalent to the requirement that $\mu_i(\+w_1) = \mu_i(\+w_2)$ for any two treatment vectors $\+w_1, \+w_2 \in \cW$ in which the $i$-th components of $\+w_1$ and $\+w_2$ agree.  \citet{manski2013identification} refers to this formulation as \emph{individualistic treatment response} (ITR).

The exposure model most commonly used for local interference is the \emph{neighborhood treatment response} (NTR) assumption, which given a graph $G$, posits that $\mu_i(\+w_1) = \mu_i(\+w_2)$ whenever $\+w_1$ and $\+w_2$ agree in all components $j$ such that $j \in \cN_i \cup \{i\}$.  In other words, NTR assumes that $Y_i$ depends on unit $i$'s own treatment and possibly any other unit in its neighborhood $\cN_i$, but that it does not respond to changes in the treatments of any units outside of its immediate neighborhood.  For the purposes of estimating the global treatment effect, one may use \emph{fractional $q$-NTR}, where given a threshold parameter $q \in (0.5, 1]$, $q$-NTR assumes that a unit is effectively in global treatment if at least a fraction $q$ of its neighbors are assigned to treatment, and similarly for global control.  NTR is thus a graph analog of partial interference for groups and $q$-NTR is a corresponding version of stratified interference.  The threshold $q$ is a tuning parameter; larger values of $q$ result in less bias due to interference, but greater variance because there are fewer units available for estimation.  \citet{eckles2017design} provide some theoretical results for characterizing the amount of bias reduction.  There is not much guidance for selecting $q$ to manage this bias-variance tradeoff; \citet{eckles2017design} uses $q = 0.75$.

\citet{aronow2017estimating} study the behavior of inverse propensity weighted (IPW) estimators based on a well-specified exposure model.  Toward this end, let 
\begin{align*}
E_i^{(\+1)} &= \one\left\{\frac{1}{d_i} \sum_{j \in \cN_i} W_j \geq q\right\} \\
E_i^{(\+0)} &= \one\left\{\frac{1}{d_i} \sum_{j \in \cN_i} W_j \leq 1 - q \right\}
\end{align*}
be the events that unit $i$ is $q$-NTR exposed to global treatment and $q$-NTR exposed to global control, respectively.  Let their expectations be denoted by
\[\pi_i^{(\+1)} = \E(E_i^{(\+1)}), \qquad \pi_i^{(\+0)} = \E(E_i^{(\+0)}),\]
which represent the \emph{propensity scores} for unit $i$ being exposed to the global potential outcome conditions.
Then the inverse propensity weighted estimators under consideration are defined as
\begin{align}
\HT &= \frac{1}{n} \sumin \left[\frac{E_i^{(\+1)} Y_i}{\pi_i^{(\+1)}} - \frac{E_i^{(\+0)} Y_i}{\pi_i^{(\+0)}}\right] \nonumber \\
\hajek &= \left(\sumin \frac{E_i^{(\+1)}}{\pi_i^{(\+1)}}\right)^{-1}\sumin \frac{E_i^{(\+1)} Y_i}{\pi_i^{(\+1)}} - \left(\sumin \frac{E_i^{(\+0)}}{\pi_i^{(\+0)}}\right)^{-1}\sumin\frac{E_i^{(\+0)} Y_i}{\pi_i^{(\+0)}} \label{eqn:hajek}
\end{align}

The estimator $\HT$ is the Horvitz-Thompson estimator~\citep{horvitz1952generalization}, and $\hajek$ is the H\'ajek estimator~\citep{hajek1971comment}; these names stem from the survey sampling literature and are commonly used in the interference literature.  In the importance sampling and off-policy evaluation literatures, analogs of $\HT$ and $\hajek$ are known as unnormalized and self-normalized importance sampling estimators, respectively.  In the finite potential outcomes framework The Horvitz-Thompson estimator is unbiased under the experimental design distribution, but suffers from excessive variance when the probabilities of global exposure are small, as is usually the case.  The H\'ajek estimator, which forces the weights to sum to one and is thus interpretable as a difference of weighted within-group means, incurs a small amount of finite sample bias but is asymptotically unbiased, and is nearly always preferable to $\HT$.  For our simulations we will therefore avoid using $\HT$.

One of the main insights in the exposure modeling framework developed by~\citet{aronow2017estimating} is that even if the initial treatment assignment probability $\pi$ is constant across units, the global treatment propensity scores need not be; indeed, $\pi_i^{(\+1)}$ and $\pi_i^{(\+0)}$ depend on the network structure and choice of exposure model.  Therefore inverse propensity weighting is needed to produce unbiased (or consistent) estimators for contrasts between exposures even in a Bernoulli randomized design.  

Given a design and a (simple enough) exposure model, the propensities can be calculated exactly.  If the treatments are assigned according to independent Bernoulli coin flips, the exact exposure probabilities are expressed straightforwardly using the binomial distribution function.  That is, for treatment probability $\pi = \P(W_i = 1)$ and degree $d_i$, the probability of unit $i$ being $q$-NTR exposed to global treatment is
\begin{equation}
\label{eqn:propensity1}
\pi_i^{(\+1)} = \pi (1 - F_{d_i, \pi}(\lfloor d_iq \rfloor)),
\end{equation}
where 
\[F_{n,p}(k) = \sum_{j=0}^k \binom{n}{j} p^j(1 - p)^{n-j}\]
is the distribution function of a Binomial$(n, p)$ random variable.
Similarly, the probability that unit $i$ is $q$-NTR exposed to global control is
\begin{equation}
\label{eqn:propensity0}
\pi_i^{(\+0)} = (1 - \pi) F_{d_i, \pi}(\lfloor d_i(1 - q) \rfloor).
\end{equation}
In a cluster randomized design, exposure probabilities for fractional neighborhood exposure can be computed using a dynamic program~\citep{ugander2013graph}.  

A further comment on the propensity scores $\pi_i^{(\+1)}$ and $\pi_i^{(\+0)}$ is necessary.  Importantly, these propensity scores are exact only to the extent to which the exposure model is correct.  Thus, when the exposure model is unknown, these propensities scores should be viewed as \emph{estimated} propensities, in which case even small estimation errors in the propensities can lead to large estimation errors in their inverses.  It is therefore the case that $\HT$ and $\hajek$ can suffer from the same high-variance problems as IPW estimators based on a fitted propensity model used in observational studies, even if the exposure model is only mildly misspecified.


In our simulations we use the H\'ajek estimator, $\hajek$, defined by equation~\eqref{eqn:hajek} and the $q$-NTR exposure probabilities~\eqref{eqn:propensity1} and~\eqref{eqn:propensity0}.   We fix $q = 0.75$, which is the same threshold used in~\citet{eckles2017design}.  For the other values of $q$ that we tried, performance was roughly on par with or worse than $q = 0.75$.

\subsection{Variance estimates in a linear model}
\label{sec:sim-basic}
We first run a basic simulation in which we compute estimates, variances and variance estimates in an ordinary linear model.  We consider two \hl{features},
\[X_{1,i} = \frac{1}{d_i} \sum_{j \in \cN_i} W_j,\]
the proportion of treated neighbors, and 
\[X_{2,i} = \sum_{j \in \cN_i} W_j,\]
the number of treated neighbors.  It is conceivable that $Y_i$ may depend on both of these features.   
Let the data-generating process for $Y_i$ be as in Model~\ref{model:linear}; that is, the mean function for $Y_i$ is linear in $X_i = (X_{1,i}, X_{2,i})$, given parameters $\alpha_w \in \R$ and $\beta_w = (\beta_{w,1}, \beta_{w,2}) \in \R^2$ for $w = 0, 1$.  We simulate $\ep_i \sim N(0, \sigma^2)$.  

Let $\bar d = n^{-1} \sumin d_i$ be the average degree of $G$.  Then the true global treatment effect is
\[\alpha_1 - \alpha_0 + \beta_{1,1} + \bar d\beta_{1,2}.\]

We fix $\alpha_1 = 1$ and $\alpha_0 = 0$, so that the direct effect is $1$.  We fix the noise variance at $\sigma^2 = 1$.  We vary the ``proportion'' coordinate of $\beta_0$ in $\{0, 0.1\}$, the ``number'' coordinate of $\beta_0$ in $\{0, 0.01\}$, the \hl{``proportion'' coordinate of $\beta_1$} in $\{0, 0.2\}$, and the  ``number'' coordinate of $\beta_1$ in $\{0, 0.05\}$, giving $16$ total parameter configurations.  SUTVA holds when $\beta_0 = \beta_1 = 0$.

We use equation~\eqref{eqn:ols-varest} to estimate the variance of the adjusted estimator, using $200$ bootstrap samples from the \hl{feature} distribution to calculate the inverse covariance matrices.  We also compute the difference-in-means (DM) estimator for comparison purposes, for which we use the standard Neyman conservative variance estimate
\[\frac{S_0^2}{N_0} + \frac{S_1^2}{N_1},\]
where $S_0^2$ and $S_1^2$ are the within-group sample variances.  We compute confidence intervals based on Gaussian quantiles for a 90\% nominal coverage rate.

We then run $1000$ simulated experiments, sampling a new treatment vector $\+W$ and computing the two estimators each time.  The results are shown in Table~\ref{table:basic-sim}.  The bias of the DM estimator increases with greater departures from SUTVA, and confidence intervals for that estimator are only theoretically valid under SUTVA (the first row in Table~\ref{table:basic-sim}).  Otherwise, the confidence intervals are anticonservative, both due to bias of the DM estimator and due to invalidity of the Neyman variance estimate, which assumes fixed potential outcomes.  On the other hand, the adjustment estimator is unbiased and has valid coverage for all parameter configurations.

\begin{table}[ht]
\centering
\begin{tabular}{rrr|rr|rr|rr|rr}
  \toprule
  \multicolumn{3}{c|}{Parameters} & \multicolumn{2}{c|}{Bias} & \multicolumn{2}{c|}{SE} & \multicolumn{2}{c|}{SE Ratio} & \multicolumn{2}{c}{Coverage rate}\\
$\beta_0$ & $\beta_1$ & $\tau$ & DM & adj & DM & adj & DM & adj & DM & adj \\ 
  \midrule
(0, 0) & (0, 0) &  1 & 0.007 & -0.013 & 0.074 & 1.149 & 0.982 & 1.031 & \textbf{0.891} & \textbf{0.913} \\ 
  (0, 0.01) & (0, 0) &  1 & -0.006 & -0.028 & 0.072 & 1.189 & 1.004 & 0.996 & \textbf{0.899} & \textbf{0.901} \\ 
  (0.1, 0) & (0, 0) &  1 & -0.053 & 0.027 & 0.072 & 1.151 & 1.004 & 1.029 & 0.808 & \textbf{0.919} \\ 
  (0.1, 0.01) & (0, 0) &  1 & -0.052 & 0.017 & 0.074 & 1.155 & 0.973 & 1.025 & 0.801 & \textbf{0.906} \\ 
  (0, 0) & (0, 0.05) &  1.05 & -0.025 & 0.005 & 0.073 & 1.211 & 0.990 & 0.976 & 0.866 & \textbf{0.882} \\ 
  (0, 0.01) & (0, 0.05) &  1.05 & -0.026 & 0.005 & 0.070 & 1.173 & 1.036 & 1.010 & \textbf{0.894} & \textbf{0.909} \\ 
  (0.1, 0) & (0, 0.05) &  1.05 & -0.075 & 0.058 & 0.075 & 1.232 & 0.960 & 0.961 & 0.707 & \textbf{0.884} \\ 
  (0.1, 0.01) & (0, 0.05) &  1.05 & -0.078 & -0.019 & 0.073 & 1.151 & 0.996 & 1.031 & 0.699 & \textbf{0.912} \\ 
  (0, 0) & (0.2, 0) &  1.2 & -0.097 & 0.058 & 0.073 & 1.168 & 0.993 & 1.014 & 0.615 & \textbf{0.910} \\ 
  (0, 0.01) & (0.2, 0) &  1.2 & -0.104 & 0.007 & 0.073 & 1.142 & 0.999 & 1.039 & 0.577 & \textbf{0.916} \\ 
  (0.1, 0) & (0.2, 0) &  1.2 & -0.151 & 0.002 & 0.073 & 1.197 & 0.991 & 0.988 & 0.334 & \textbf{0.892} \\ 
  (0.1, 0.01) & (0.2, 0) &  1.2 & -0.152 & 0.044 & 0.072 & 1.208 & 1.014 & 0.981 & 0.315 & \textbf{0.894} \\ 
  (0, 0) & (0.2, 0.05) &  1.25 & -0.125 & -0.054 & 0.072 & 1.154 & 1.003 & 1.025 & 0.476 & \textbf{0.908} \\ 
  (0, 0.01) & (0.2, 0.05) &  1.25 & -0.130 & -0.014 & 0.070 & 1.149 & 1.029 & 1.031 & 0.446 & \textbf{0.920} \\ 
  (0.1, 0) & (0.2, 0.05) &  1.25 & -0.174 & 0.016 & 0.074 & 1.206 & 0.983 & 0.982 & 0.232 & \textbf{0.894} \\ 
  (0.1, 0.01) & (0.2, 0.05) &  1.25 & -0.182 & -0.014 & 0.074 & 1.172 & 0.985 & 1.012 & 0.194 & \textbf{0.903} \\ 
   \bottomrule
\end{tabular}
\caption{Results of the basic simulation setup from Section~\ref{sec:sim-basic}, showing bias, true standard error, ratio of estimated standard error to true standard error, and coverage rate of 90\% nominal Gaussian confidence interval.  Coverage rates which fall within a 99\% one-sided interval of the nominal coverage rate (that is, coverage rates above $0.9 - 2.326 \sqrt{0.9 \times 0.1 / 1000} \approx 0.878$) are \textbf{bolded}.}
\label{table:basic-sim}
\end{table}

\subsection{Estimator weights}
\label{sec:sim-weights}
Both the OLS adjustment estimator and the H\'ajek estimator are linear reweighting estimators.  The OLS weights are given by equations~\eqref{eqn:ols-weight0} and~\eqref{eqn:ols-weight1}, and the H\'ajek weights are implied by the definition of the H\'ajek estimator in equation~\eqref{eqn:hajek}.  Both depend on only the network structure, treatment assignment, and exposure model or choice of \hl{features}, but not on the realized outcome variable.  The weights for a single Bernoulli$(0.5)$ draw of the treatment vector $\+W$ for the Caltech graph are displayed in Figure~\ref{fig:weights}, assuming that the H\'ajek estimator is to be constructed under the $q$-NTR exposure condition for $q = 0.75$, and the OLS estimator uses the fraction of treated neighbors as the only adjustment variable.  
We see that the H\'ajek estimator trusts a few select observations to be representative of the global exposure conditions.  A graph cluster randomized design would increase the number of units used in the H\'ajek estimator.  The OLS estimator, on the other hand, gives all units non-zero weight.  Some units that are in the treatment group but are surrounded by control individuals are treated as diagnostic for the control mean and vice versa, which is a reasonable thing to do if the linear model is true.

\begin{figure}[ht]
\centering
\includegraphics[width=0.8\textwidth]{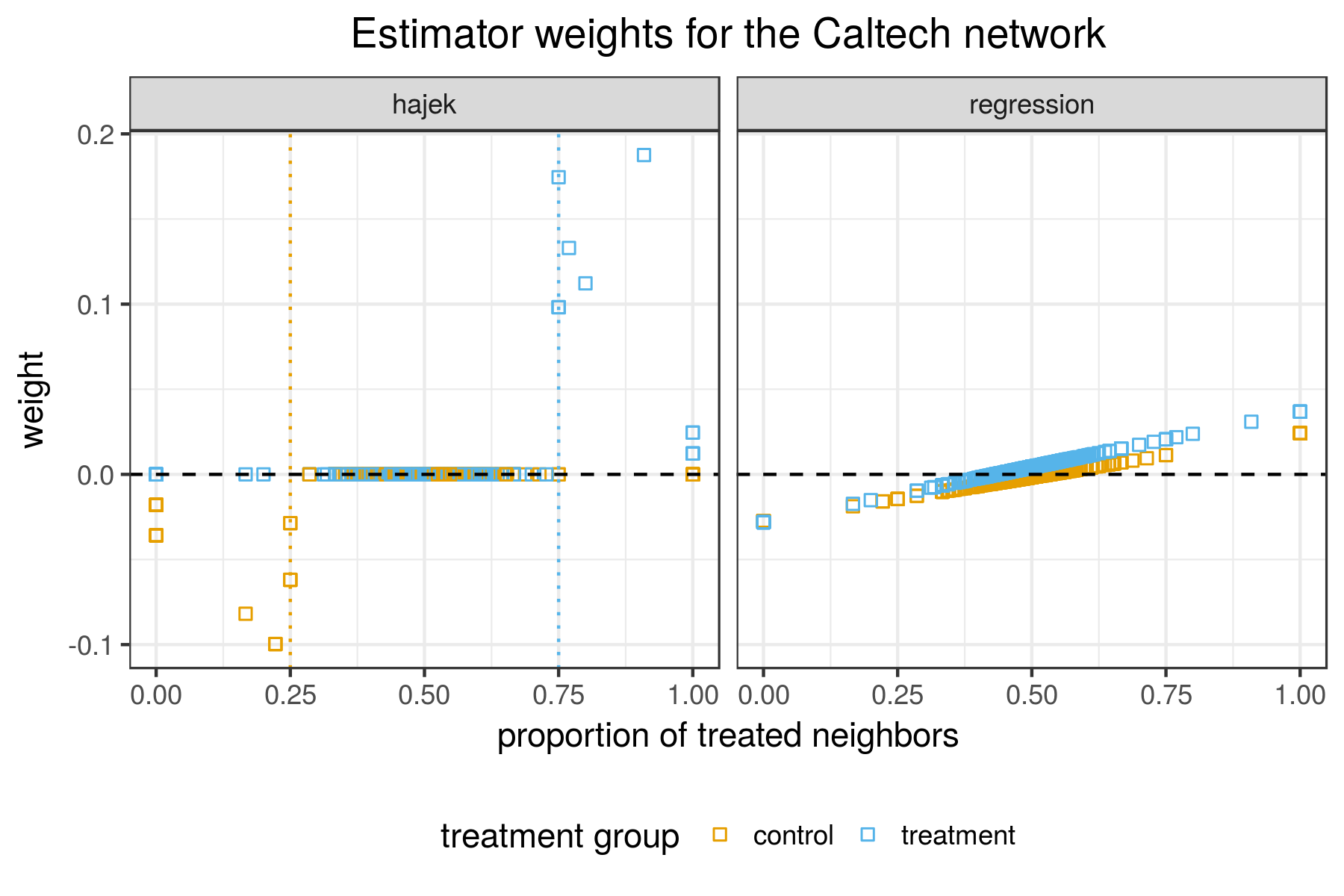}
\caption{Estimator weights for the case where the only \hl{feature} is the proportion of treated neighbors.  (left) The H\'ajek estimator selects a few individuals from treatment and control and takes a weighted average of those individuals with weights determined by exposure probabilities.  Vertical dotted lines are the thresholds used for selecting observations.  (right) The regression estimator takes a more democratic approach, giving all units non-zero weight.}
\label{fig:weights}
\end{figure}

\subsection{Dynamic linear-in-means}

Here we replicate portions of the simulation experiments conducted by~\citet{eckles2017design}.  That paper uses a discrete-time dynamic model, which can be viewed as a noisy best-response model~\citep{blume1995statistical}, in which individuals observe and respond to the behaviors of their peers, using that information to guide their actions in the following time period.  Given responses $Y_{i,t-1}$ for time period $t-1$, let 
\[Z_{i,t-1} = \frac{1}{d_i} \sum_{j \in \cN_i} Y_{i,t-1}, \]
a time-varying version of $Z_i$ defined in equation~\eqref{eqn:endogenous-mean}, which represents the average behavior of unit $i$'s neighbors at time $t - 1$. 
Then we model
\begin{equation}
\label{eqn:dynamic-lim}
Y_{i,t} = \alpha + \beta W_i + \gamma \hl{Z_{i,t-1}} + \ep_{i,t}.
\end{equation}
The noise is taken to be $\ep_{i,t} \sim N(0, \sigma^2)$, which is independent and homoscedastic across time and individuals.  
\citet{eckles2017design} add an additional thresholding step that transforms equation~\eqref{eqn:dynamic-lim} into a probit model and $Y$ into a binary outcome variable, but here we study the non-thresholded case which is closer to the original linear-in-means model specified by~\cite{manski1993identification}.  
Starting from initial values $Y_{i,0} = 0$, the process is run up to a maximum time $T$ and then the final outcomes are taken to be $Y_i = Y_{i,T}$.  The choice of $T$, along with the strength of the spillover effect $\gamma$, governs the amount of interference.  If $T$ is larger than the diameter of the graph, then the interference pattern is fully dense, and no exposure model holds.

We construct two different adjustment variables.  First, let
\[X_{1,i} = \frac{1}{d_i} \sum_{j \in \cN_i} W_j,\]
the proportion of treated neighbors.
Now let
\[\cN_i^{(2)} = \{k  \in [n] \setminus \{i\} : \text{there exists } j \text{ such that } A_{ij} A_{jk} = 1\}\]
be the two-step neighborhood of unit $i$.  Then define
\[X_{2,i} = \frac{1}{|\cN_i^{(2)}|} \sum_{k \in \cN_i^{(2)}} A_{ij}A_{jk} W_k,\]
the proportion of individuals belonging to $\cN_i^{(2)}$ who are treated.  (Note that unit $i$ itself does \emph{not} belong to its own two-step neighborhood.)

We use a small-world network~\citep{watts1998collective}, which is the random graph model used in the simulations by~\citet{eckles2017design}, with $n = 1000$ vertices, initial neighborhood size $10$, and rewiring probability $0.1$.  We also run our simulation on the empirical Caltech network.

As in~\citet{eckles2017design}, we compute the ``true'' global treatment effects by Monte Carlo simulation.  For every parameter configuration we sample 5000 instances of the response vector under global exposure to treatment $\+W = \+1$, and 5000 instances of the response vector under global exposure to control $\+W = \+0$, and then average the resulting difference in response means.  For the response model, we fix the intercept at $\alpha = 0$ and the direct effect at $\beta = 1$.  We vary the spillover effect $\gamma \in \{0, 0.25, 0.5, 0.75, 1\}$ and the maximum number of time steps $T \in \{2, 4\}$.  Larger values of $\gamma$ and $T$ indicate more interference.  We also use two different levels for the noise standard deviation, $\sigma \in \{1, 3\}$.

We consider two versions of the linear adjustment estimator defined in equation~\eqref{eqn:linear-tau-hat}, one that adjusts for $X_{1,i}$ only, and one that adjusts for both $X_{1,i}$ and $X_{2,i}$.  The first model adjusts for one-step neighborhood information, whereas the second model adjusts for both one- and two-step neighborhood information.  We compare to the difference-in-means estimator and the H\'ajek estimator with $q = 0.75$ fractional NTR exposure.

We emphasize that the all of the estimators that we consider are misspecified under the data generating process that we use in this simulation.  For $T \geq 2$, local neighborhood exposure fails, so the propensity scores used in the H\'ajek estimator do not align with the true propensity scores.  Our adjustment estimators are also misspecified for $T \geq 2$; not only is the linear model misspecified, but the residuals are neither independent nor exogenous, violating Assumption~\ref{asm:errors}.

The results are displayed in Figure~\ref{fig:lim-results}.  We see that the two OLS adjustment estimators are uniformly better at bias reduction than the H\'ajek estimator.  The two-step adjustment is nearly unbiased even though it is misspecified, even in the presence of strong spillover effects.  This is because interference is dissipating exponentially, so that units don't really respond to the behavior of individuals that are distance $3$ or $4$ away.  The two-step adjustment has higher variance than the one-step adjustment because it involves fitting a more complex model.  Furthermore, estimators appear to have more trouble handling the real-world network structure of the Caltech network, compared to the artificial small-world network.

\begin{figure}[!ht]
\centering
\includegraphics[width=\textwidth]{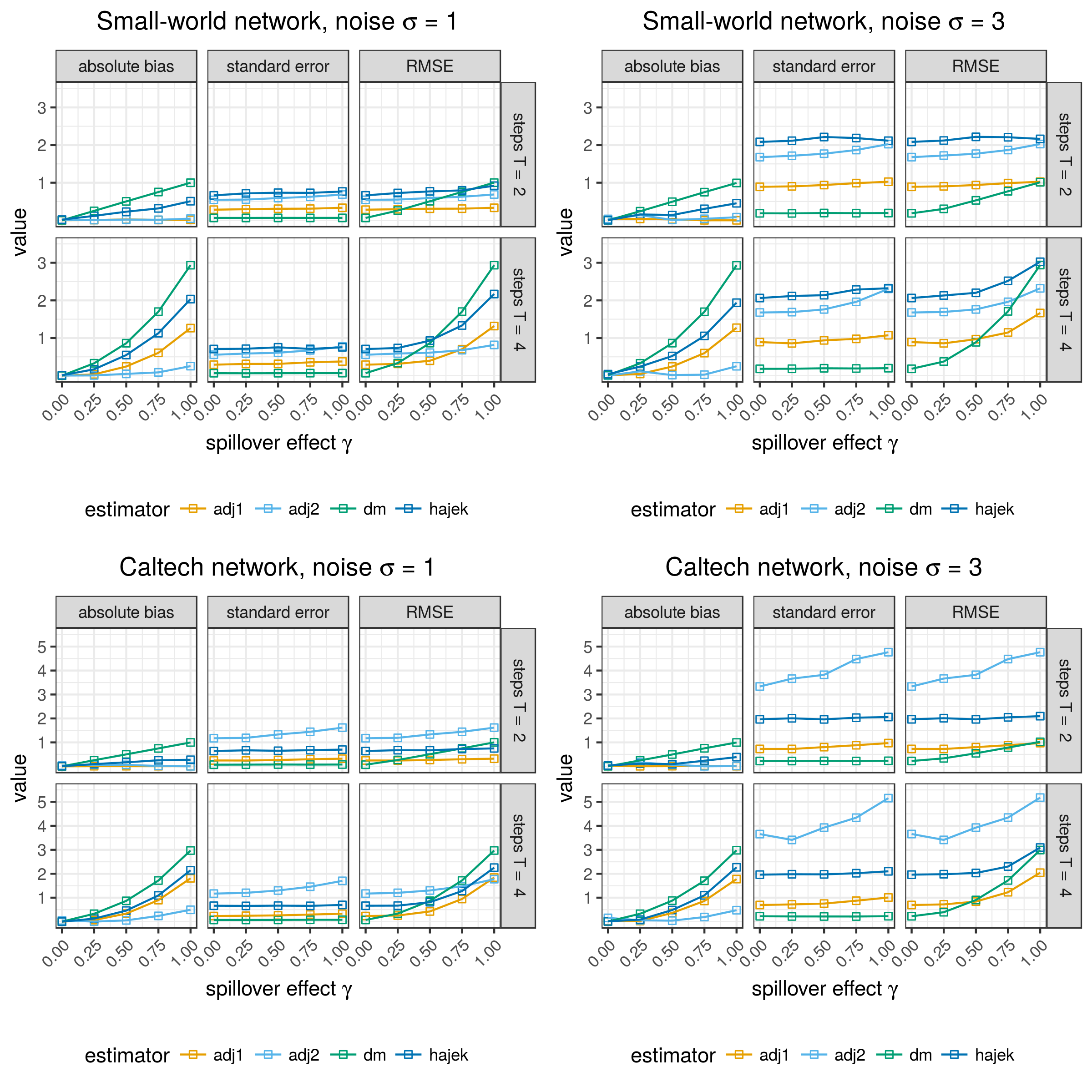}
\caption{Results for linear-in-means simulation.  \texttt{dm} is the difference-in-means estimator, \texttt{hajek} is the H\'ajek estimator, \texttt{adj1} is adjustment based on a one-step neighborhood, and \texttt{adj2} is adjustment based on a two-step neighborhood.}
\label{fig:lim-results}
\end{figure}

The difference-in-means estimator outperforms the adjustment estimators in regimes of weak interference, which is expected since difference-in-means is the best that can be done under correct specification of SUTVA.  In terms of RMSE the H\'ajek estimator sometimes outperforms the two-step adjustment estimator because of large variance.  However, if the main goal is robustness to interference, then unconfounded estimation coupled with valid confidence intervals is likely the priority over optimizing an error metric such as RMSE.  In this case, since the H\'ajek estimator neither achieve sufficient bias reduction nor provide correct coverage, it has no real advantage over the adjustment estimators.

Figure~\ref{fig:lim-coverage} displays the coverage rates obtain from variance estimates using equation~\eqref{eqn:ols-varest} under the dynamic treatment response setup.  The coverage is not always correct due to misspecification, especially for \texttt{adj1}.  We see that coverage rates for \texttt{adj2} are often conservative even though it too is misspecified.  We note that standard variance estimators for the difference-in-means estimator and those derived in~\cite{aronow2017estimating} for the H\'ajek estimator also would fail here because they rely on correct specification of SUTVA and an exposure model, respectively.  In short, we are plagued with the same difficulties that beset attempting to do valid inference in observational studies when we do not know whether unconfoundedness holds.

\begin{figure}[!ht]
\centering
\includegraphics[width=\textwidth]{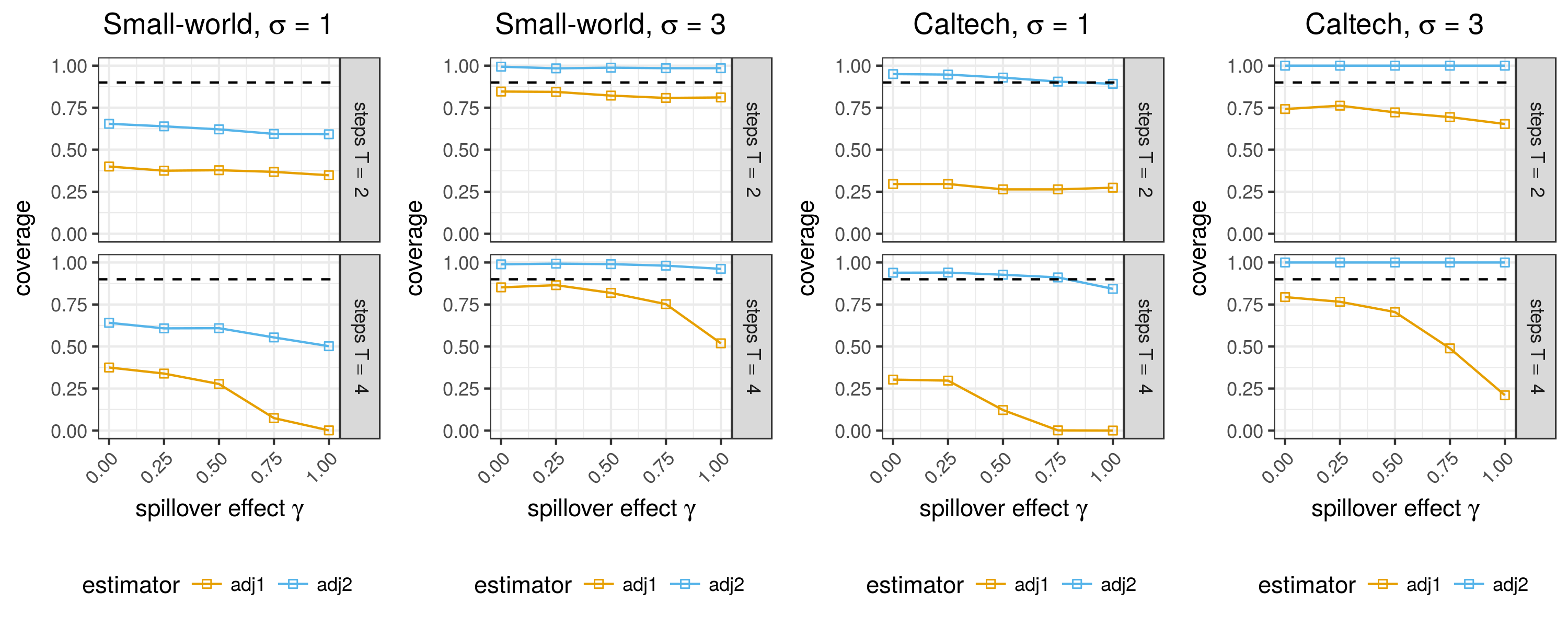}
\caption{Coverage rates for 90\% nominal interval.}
\label{fig:lim-coverage}
\end{figure}





\subsection{Average + aggregate peer effects}
In this example we consider a response model in which individuals respond partially to the \emph{average} behavior of their peers and partially to the \emph{aggregate} behavior of their peers.  Let
\[X_i^{\text{frac}} = \frac{1}{d_i} \sum_{j \in \cN_i} W_j\]
be the fraction of treated neighbors and
\[X_i^{\text{num}} = \sum_{j \in \cN_i} W_j\]
be the number of treated neighbors.  $X_i^{\text{frac}}$ captures a notion of fractional neighborhood exposure and $X_i^{\text{num}}$ captures a notion of absolute neighborhood exposure.  It seems reasonable that both of these features may contribute interference.  In order to use an exposure model estimator one would need to focus on either fractional exposure or absolute exposure, or otherwise define a more complicated exposure model, but our adjustments easily handle both features.

We consider the following response function:
\[Y_i = -5 + 2 (2 + E_i) W_i + 0.03X_i^{\text{frac}} + \frac{1}{1 + 0.001 e^{-0.03 (X_i^{\text{num}} - 300)}} + \frac{10}{3 + e^{-8(X_i^{\text{frac}} - 0.4)}} + \ep_i,\]
where $E_i \sim N(0, 2)$ introduces heterogeneity into the direct effect and $\ep_i \sim N(0, 1)$ is homoscedastic noise.
This function captures a possible way in which individuals could respond nonlinearly to their peer exposures through $X_i^{\text{frac}}$ and $X_i^{\text{num}}$.  Figure~\ref{fig:nonlinear-response} plots a single draw of this response on individuals from the Stanford network.  The continuous response exhibits a logistic dependence on both \hl{features}.  We see that individuals with less than half of their neighbors exposed to the treatment condition experience a steadily increasing peer effect as the proportion of treated neighbors increases.  For individuals with more than half of their neighbors exposed to the treatment condition, the effect is nearly constant across values of $X_i^{\text{frac}}$, capturing the idea that after a certain threshold observing additional peer exposures doesn't add much.  For $X_i^{\text{num}}$, we see that a small number of treated neighbors essentially contributes no interference, but once a large number of neighbors are exposed to treatment this has a measurable impact on the response.  We also see that there is a noticeable bump around $X_i^{\text{frac}} = 0.5$; this is because individuals with peers nearly equally assigned to the two groups are more likely to have high degree.  The model extends the idea of neighborhood exposure to capture the intuition that having a high proportion of treated neighbors is evidence for being subject to interference, but such evidence is stronger when the individual in question has many friends and not just one or two friends.  The true treatment effect is $\tau = 6.336$, which was computed using 2000 Monte Carlo draws each of global treatment and global control.

\begin{figure}[!ht]
\centering
\includegraphics[width=\textwidth]{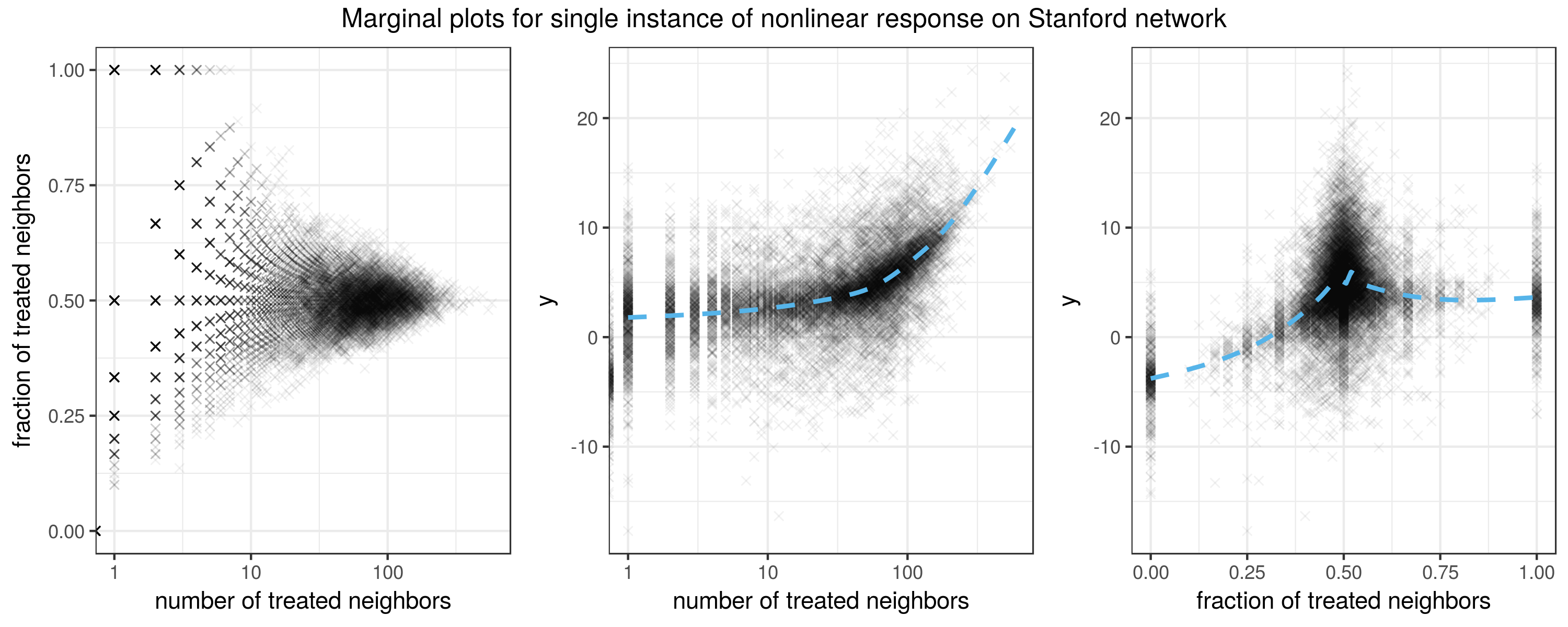}
\caption{One draw of the \hl{features} and response for the nonlinear setup. The left panel shows the relationship between the two \hl{features}, and the right two panels show the relationship of the response with each \hl{covariate}.  The horizontal axis for ``number of treated neighbors'' ($X_i^{\text{num}}$) is on a logarithmic scale.  A local linear regression, for exploratory purposes, is plotted in blue.}
\label{fig:nonlinear-response}
\end{figure}

In our experience larger populations seem to be needed for fitting the more complex, nonlinear functions, so we work with the Stanford network which has 11586 nodes.  We predict the response surfaces using a generalized additive model (GAM)~\citep{hastie1986generalized}, which is easy and fast to fit in \texttt{R}, but other methods such as local regression or random forests could of course be used instead.  We split the dataset into $K = 2$ folds, and within each fold, train a GAM separately in the treatment and control groups for a total of 4 fitted models.  The models are then used to obtain predicted responses on the held-out fold.  Standard errors were computed via the bootstrap as described in Section~\ref{sec:nonparametric-varest}, using 50 bootstrap replications.

We compare to the difference-in-means estimator, the H\'ajek estimator using a threshold of $q = 0.75$ on the $X_i^{\text{frac}}$ variable, and the OLS adjustment.   The results are displayed in Table~\ref{table:nonlinear-results}. The DM estimator exhibits the most bias, as it does not adjust for any sort of interference.  The H\'ajek estimator removes some bias, but because it is based on a fractional exposure model it is unable to respond to the effect of having a high treated degree.  Both the OLS and GAM estimators remove about 95\% of the bias.  The GAM adjustment does only slightly better than OLS; for this setup what matters most is adjusting for both axes of the interference \hl{statistic}, and the added flexibility provided by the GAM does not seem to be crucial.  We note also that average bootstrapped standard error is 1.076 times greater than the true standard error, suggesting that confidence intervals built on this standard error will have the approximately correct length. 

\begin{table}[ht]
\centering
\begin{tabular}{lrrr}
  \toprule
estimator & estimate & absolute bias (\%) & SE (ratio) \\ 
  \midrule
DM & -0.002 & 6.339 (100\%) & 0.077 (---)\\ 
  H\'ajek & 2.653 & 3.683 (58.1\%) & 1.601 (---) \\ 
  OLS & 6.683 & 0.347 (5.5\%) & 0.252 (0.942) \\ 
  GAM & 6.655 & 0.319 (5.0\%) & 0.246 (1.076) \\ 
   \bottomrule
\end{tabular}
\caption{Nonlinear simulation results.  The bias column displays the absolute and relative bias from the truth $\tau = 6.336$.  The SE column displays the true standard error over 200 simulation replications, and for the adjustment estimators we display in parentheses the ratio of the estimated standard error to the true standard error.}
\label{table:nonlinear-results}
\end{table}

\section{Application to a farmer's insurance experiment}
\label{sec:applications}

In this section we apply our methods to a field experiment conducted on individuals in 185 villages in rural China~\citep{cai2015social}.  The purpose of the study was to quantify the network (spillover) effects of certain information sessions for a farmer's weather insurance product on the eventual adoption of that product.  Though they do not frame their approach explicitly in the language of exposure models as in~\citep{aronow2017estimating}, the estimands that are implied by the regression coefficients in the models that they use in that paper can be thought of as contrasts between exposures in an appropriately-defined exposure model.  The authors did not consider estimating the global treatment effect; our proposed methods essentially allow us to perform an off-policy analysis of that estimand.

In the original field experiment, the researchers consider four treatment groups obtained by assigning villagers to either a simple or intensive information session in one of two rounds that were held three days apart.  Here, for simplicity, we ignore the temporal distinction between the two rounds and consider a villager to be treated if they were exposed to either of the two intensive sessions.\footnote{According to~\citet{cai2015social} the treatment groups in the study are stratified by household size and farm size, but it is not clear from the data if and how exactly this was done, so for simplicity we analyze the experiment as if it were an unstratified, Bernoulli randomized experiment.}  The outcome variable is a binary indicator for whether the villager decided to purchase weather insurance.

We drop all villagers that were missing information about the treatment or the response, as well as villages lacking network information.  Though the study was conducted in separate villages (for the purpose of administering the insurance information sessions), we combine all of the villagers into one large graph $G$.  The network has 4,382 nodes and 17,069 edges.  Because some social connections exist across villages, the villages do not partition exactly into separate connected components; our graph $G$ has 36 connected components.  The summary statistics for the processed dataset are given in Table~\ref{table:cai-summary}.

\begin{table}[ht]
\centering
\begin{tabular}{r|r}
\toprule 
number of nodes & 4832 \\
number of edges & 17069 \\
number (\%) treated & 2406 (49.8\%) \\
average takeup (mean response) & 44.6\% \\
\bottomrule
\end{tabular}
\caption{Summary statistics for the \citep{cai2015social} dataset.}
\label{table:cai-summary}
\end{table}

Now let $\cN_i$ and $\cN_i^{(2)}$ be the one- and two-step neighborhoods for unit $i$, as we have denoted previously.  We construct four variables from the graph: the fraction of units in $\cN_i$ who are treated (\texttt{frac1}), the fraction of units in $\cN_i^{(2)}$ who are treated (\texttt{frac2}), the number of units in $\cN_i$ who are treated (\texttt{num1}), and the number of units in $\cN_i^{(2)}$ who are treated (\texttt{num2}).  Figure~\ref{fig:cai-scatter} displays the scatterplot matrix for these four variables as well as the response.  As might be expected, these four variables are positively correlated with each other, and each is (weakly) positively correlated with the response variable.  This correlation with the response suggests that these variables may be useful for adjustment.
\begin{figure}[!ht]
\centering
\includegraphics[width=\textwidth]{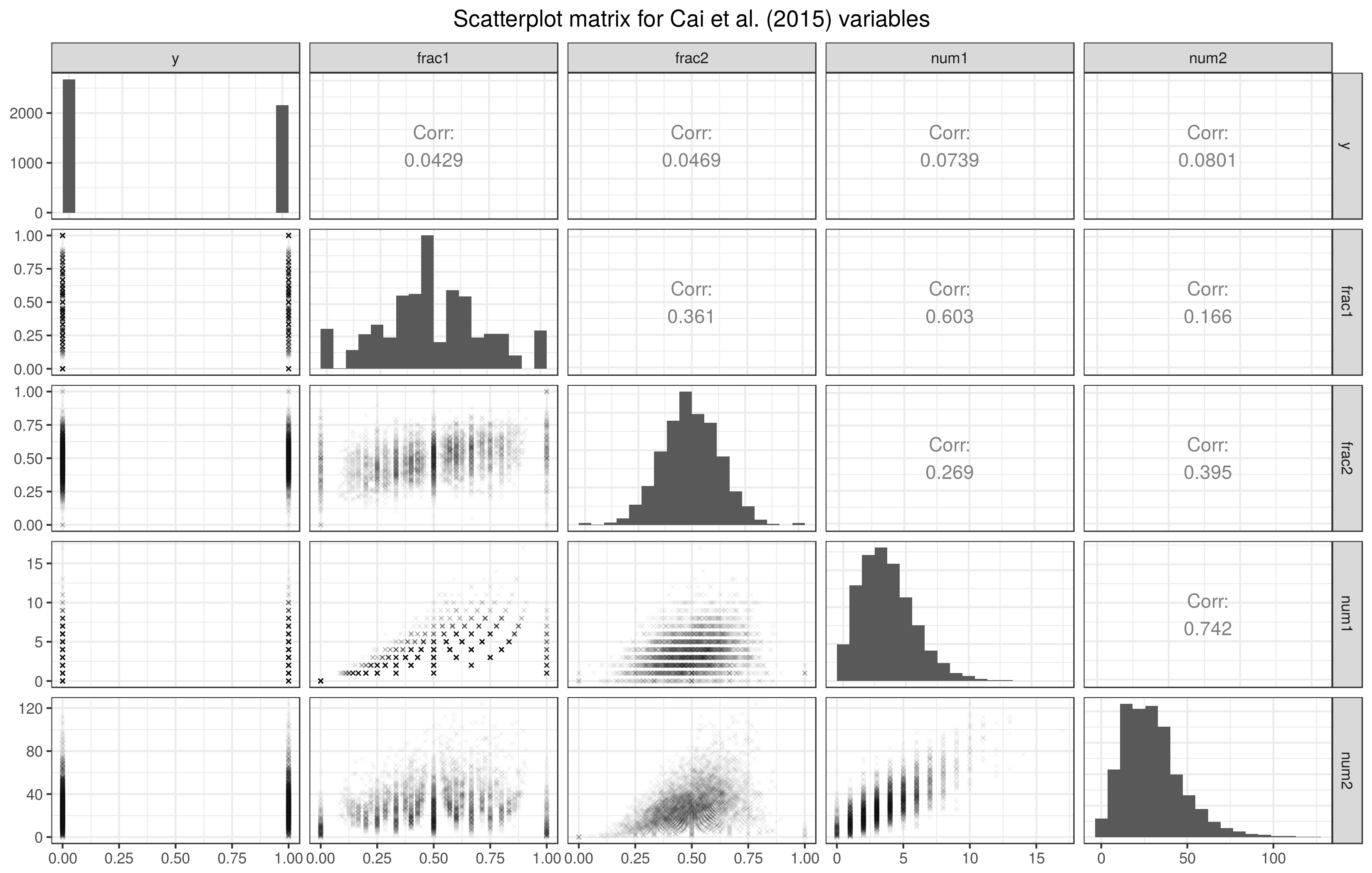}
\caption{Scatterplot matrix for the variables used in the~\citet{cai2015social} analysis.}
\label{fig:cai-scatter}
\end{figure}


We compute the OLS adjusted estimator as well as an adjustment estimator that used predictions from a logistic regression with $K = 5$ folds.  We construct standard errors using the variance estimator given by equation~\eqref{eqn:ols-varest} in the OLS case, and the parametric bootstrap variance estimator described in Section~\ref{sec:nonparametric-varest} with 200 bootstrap replications for the logistic regression case.  We compare to the difference-in-means estimator and H\'ajek estimators based on thresholding on the \texttt{frac1} and \texttt{frac2} variables with $q = 0.75$.  

The estimates are displayed in Table~\ref{table:cai-results}.  Considering the strong positive spillover effects discovered by~\citet{cai2015social}, the difference-in-means estimate of 0.0774 is likely to be an underestimate of the true global treatment effect.  The H\'ajek estimators produce estimates of 0.1630 (one-step fractional NTR) and 0.1672 (two-step fractional NTR).  Though we do not know the truth, it may make us nervous that these estimates are more than twice the magnitude of the difference-in-means estimator, which if true would suggest that magnitude of the spillover effect is larger than the magnitude of the direct effect.  The true treatment effect likely falls in between the estimates produced by difference-in-means and H\'ajek (though we have no way of knowing for sure). The OLS (0.1218) and logistic regression estimates (0.1197) are similar to each other and both within this range; an advantage they have over the H\'ajek estimators is that they incorporate information about the raw number of treated neighbors.  The standard error estimates of 0.0561 (linear regression adjustment) and 0.0559 (logistic regression adjustment) are quite wide, suggesting some caution when interpreting this result.  

Note that we have omitted computation of standard error estimates for the difference-in-means and H\'ajek estimators for several reasons.  SUTVA and the neighborhood exposure conditions both likely fail to hold, so it is unclear how we should interpret such standard errors.  Secondly, the conservative variance estimators proposed for the H\'ajek estimator~\citep[cf.\ Sections 5, 7.2,][]{aronow2017estimating} are themselves inverse propensity estimators relying on small propensities, and consequently we found them to be quite unstable.  For example, the variance estimate was often much greater than 1, which is the maximum possible variance of a $[-1,1]$-valued random variable.  Of course, we also do not know if the exogeneity assumptions hold or if other variables should be included.  In the regression analyses conducted by~\citet{cai2015social}, they also consider some other social network measures including indicator variables for varying numbers of friends and differentiation between strong and weak ties; a more sophisticated analysis here could include these features as well.

\begin{table}[ht]
\centering
\begin{tabular}{lrr}
  \hline
estimator & estimate & standard error \\ 
  \hline
DM & 0.0774 & --- \\ 
  H\'ajek 1 ($q = 0.75$) & 0.1630 & --- \\ 
  H\'ajek 2 ($q = 0.75$) & 0.1672 & --- \\
  Linear & 0.1218 & 0.0561 \\ 
  Logistic (5-fold) & 0.1197 & 0.0559 \\ 
   \hline
\end{tabular}
\caption{Estimates and standard errors for estimating the global treatment effect of intensive session on insurance adoption.}
\label{table:cai-results} 
\end{table}

\section{Discussion}
\label{sec:discussion}
We propose regression adjustments for interference in randomized experiments, which opens the world of the rich regression adjustment literature to the interference setting.  We show in simulation experiments that the adjustments can do well, and we show how to do inference under exogeneity/unconfoundedness assumptions.  Our reanalysis of the \citet{cai2015social} study shows that our approach can produce sensible estimates of the global treatment effect on real data.

There is much work to do to ensure that this approach can be reliably used in practical settings.  First, we would like to extend the methods to handle more complicated designs.  In reality a combination of design-side methods (graph clustering) and analysis-side methods (adjustment) could be the most effective approach.  \hl{It would also be useful to have a thorough understanding of the combinations of network structures and experimental designs that correspond to the mathematical assumptions (exogeneity, full-rank design) listed in this paper.}

\hl{Secondly, it is necessary to formalize the placement of the methods discussed here within the agnostic perspective to treatment effect estimation.  This would clarify the exogeneity/unconfoundedness requirement and better elucidate how interference causes a randomized experiment to behave in some ways like an observational study.  However, such assumptions are not new, and also needed to employ both standard estimators for observational studies in the SUTVA setting and exposure modeling estimators in the interference setting.}

This issue simply highlights the need for better methods that can detect interference; there are several budding possibilities here.  First, several works have proposed ways of doing sensitivity analysis for interference.  \citet{vanderweele2014interference} extend \citet{robins2000sensitivity}-style sensitivity analysis to cover some of the interference estimators studied in~\citet{hudgens2008toward}, and ~\citet{egami2017unbiased} propose using an auxiliary network to perform sensitivity analysis on estimates obtained using the primary network.  But clearly more work in this area is needed.  Second, hypothesis tests for network or spillover effects of the type developed in~\citep{aronow2012general,athey2017exact,basse2017exact}, could be informative if applied to the residuals of a fitted interference model.  Finally, one can always use more robust standard error constructions such as Eicker-Huber-White~\citep{eicker1967limit,huber1967behavior,white1980heteroskedasticity} standard errors for heteroscedasticity or cluster bootstrap methods for dependence, though if the network structure is such that graph cluster randomization is unlikely to work well, then clustered bootstrap probably won't work well either.  
It is also possible that work based on dependency central limit theorems like the ones considered in~\citet{chin2018central} could be used to develop more robust variance calculations.
Broadly, any of the above methods ideas can be applied to the residuals of an interference model.  If the bulk of interference can be captured in the mean function, then it is perhaps easier to deal with the remaining interference in the residuals.

\pagebreak


\begin{thebibliography}{30}
\providecommand{\natexlab}[1]{#1}
\providecommand{\url}[1]{\texttt{#1}}
\expandafter\ifx\csname urlstyle\endcsname\relax
  \providecommand{\doi}[1]{doi: #1}\else
  \providecommand{\doi}{doi: \begingroup \urlstyle{rm}\Url}\fi

\bibitem[Athey et~al.(2017)Athey, Imbens, and Wager]{athey2017approximate}
S.~Athey, G.~W. Imbens, and S.~Wager.
\newblock Approximate residual balancing: De-biased inference of average
  treatment effects in high dimensions.
\newblock 2017.
\newblock URL: \url{https://arxiv.org/pdf/1604.07125.pdf}.

\bibitem[Berk et~al.(2013)Berk, Pitkin, Brown, Buja, George, and
  Zhao]{berk2013covariance}
R.~Berk, E.~Pitkin, L.~Brown, A.~Buja, E.~George, and L.~Zhao.
\newblock Covariance adjustments for the analysis of randomized field
  experiments.
\newblock \emph{Evaluation Review}, 37\penalty0 (3-4):\penalty0 170--196, 2013.

\bibitem[Bloniarz et~al.(2016)Bloniarz, Liu, Zhang, Sekhon, and
  Yu]{bloniarz2016lasso}
A.~Bloniarz, H.~Liu, C.-H. Zhang, J.~S. Sekhon, and B.~Yu.
\newblock Lasso adjustments of treatment effect estimates in randomized
  experiments.
\newblock \emph{Proceedings of the National Academy of Sciences}, 113\penalty0
  (27):\penalty0 7383--7390, 2016.

\bibitem[Bramoull{\'e} et~al.(2009)Bramoull{\'e}, Djebbari, and
  Fortin]{bramoulle2009identification}
Y.~Bramoull{\'e}, H.~Djebbari, and B.~Fortin.
\newblock Identification of peer effects through social networks.
\newblock \emph{Journal of Econometrics}, 150\penalty0 (1):\penalty0 41--55,
  2009.

\bibitem[Chernozhukov et~al.(2018)Chernozhukov, Chetverikov, Demirer, Duflo,
  Hansen, Newey, and Robins]{chernozhukov2018double}
V.~Chernozhukov, D.~Chetverikov, M.~Demirer, E.~Duflo, C.~Hansen, W.~Newey, and
  J.~Robins.
\newblock Double/debiased machine learning for treatment and structural
  parameters.
\newblock \emph{The Econometrics Journal}, 21\penalty0 (1):\penalty0 C1--C68,
  2018.

\bibitem[Eckles et~al.(2017)Eckles, Karrer, and Ugander]{eckles2017design}
D.~Eckles, B.~Karrer, and J.~Ugander.
\newblock Design and analysis of experiments in networks: Reducing bias from
  interference.
\newblock \emph{Journal of Causal Inference}, 5\penalty0 (1), 2017.

\bibitem[Efron(1979)]{efron1979bootstrap}
B.~Efron.
\newblock Bootstrap methods: Another look at the jackknife.
\newblock \emph{The Annals of Statistics}, 7\penalty0 (1):\penalty0 1--26,
  1979.

\bibitem[Efron(1987)]{efron1987better}
B.~Efron.
\newblock Better bootstrap confidence intervals.
\newblock \emph{Journal of the American Statistical Association}, 82\penalty0
  (397):\penalty0 171--185, 1987.

\bibitem[Forastiere et~al.(2016)Forastiere, Airoldi, and
  Mealli]{forastiere2016identification}
L.~Forastiere, E.~M. Airoldi, and F.~Mealli.
\newblock Identification and estimation of treatment and interference effects
  in observational studies on networks.
\newblock \emph{arXiv preprint arXiv:1609.06245}, 2016.

\bibitem[Freedman(2008{\natexlab{a}})]{freedman2008regressionA}
D.~A. Freedman.
\newblock On regression adjustments to experimental data.
\newblock \emph{Advances in Applied Mathematics}, 40\penalty0 (2):\penalty0
  180--193, 2008{\natexlab{a}}.

\bibitem[Freedman(2008{\natexlab{b}})]{freedman2008regressionB}
D.~A. Freedman.
\newblock On regression adjustments in experiments with several treatments.
\newblock \emph{The annals of applied statistics}, 2\penalty0 (1):\penalty0
  176--196, 2008{\natexlab{b}}.

\bibitem[Greenland(1987)]{greenland1987interpretation}
S.~Greenland.
\newblock Interpretation and choice of effect measures in epidemiologic
  analyses.
\newblock \emph{American journal of epidemiology}, 125\penalty0 (5):\penalty0
  761--768, 1987.

\bibitem[Hudgens and Halloran(2008)]{hudgens2008toward}
M.~G. Hudgens and M.~E. Halloran.
\newblock Toward causal inference with interference.
\newblock \emph{Journal of the American Statistical Association}, 103\penalty0
  (482):\penalty0 832--842, 2008.

\bibitem[K\"unsch(1989)]{kunsch1989jackknife}
H.~R. K\"unsch.
\newblock The jackknife and the bootstrap for general stationary observations.
\newblock \emph{The Annals of Statistics}, pages 1217--1241, 1989.

\bibitem[Lin(2013)]{lin2013agnostic}
W.~Lin.
\newblock Agnostic notes on regression adjustments to experimental data:
  Reexamining {F}reedman’s critique.
\newblock \emph{The Annals of Applied Statistics}, 7\penalty0 (1):\penalty0
  295--318, 2013.

\bibitem[Manski(1993)]{manski1993identification}
C.~F. Manski.
\newblock Identification of endogenous social effects: The reflection problem.
\newblock \emph{The Review of Economic Studies}, 1993.

\bibitem[Manski(2013)]{manski2013identification}
C.~F. Manski.
\newblock Identification of treatment response with social interactions.
\newblock \emph{The Econometrics Journal}, 16\penalty0 (1), 2013.

\bibitem[Montgomery et~al.(2018)Montgomery, Nyhan, and
  Torres]{montgomery2018conditioning}
J.~M. Montgomery, B.~Nyhan, and M.~Torres.
\newblock How conditioning on posttreatment variables can ruin your experiment
  and what to do about it.
\newblock \emph{American Journal of Political Science}, 62\penalty0
  (3):\penalty0 760--775, 2018.

\bibitem[Ogburn and VanderWeele(2014)]{ogburn2014causal}
E.~L. Ogburn and T.~J. VanderWeele.
\newblock Causal diagrams for interference.
\newblock \emph{Statistical Science}, 29\penalty0 (4):\penalty0 559--578, 2014.

\bibitem[Pearl(2009)]{pearl2009causality}
J.~Pearl.
\newblock \emph{Causality}.
\newblock Cambridge University Press, 2009.

\bibitem[Robins and Greenland(1989)]{robins1989probability}
J.~Robins and S.~Greenland.
\newblock The probability of causation under a stochastic model for individual
  risk.
\newblock \emph{Biometrics}, pages 1125--1138, 1989.

\bibitem[Robins and Greenland(2000)]{robins2000causal}
J.~M. Robins and S.~Greenland.
\newblock Causal inference without counterfactuals: comment.
\newblock \emph{Journal of the American Statistical Association}, 95\penalty0
  (450):\penalty0 431--435, 2000.

\bibitem[Sussman and Airoldi(2017)]{sussman2017elements}
D.~L. Sussman and E.~M. Airoldi.
\newblock Elements of estimation theory for causal effects in the presence of
  network interference.
\newblock \emph{arXiv preprint arXiv:1702.03578}, 2017.

\bibitem[Traud et~al.(2011)Traud, Kelsic, Mucha, and
  Porter]{traud2011comparing}
A.~L. Traud, E.~D. Kelsic, P.~J. Mucha, and M.~A. Porter.
\newblock Comparing community structure to characteristics in online collegiate
  social networks.
\newblock \emph{SIAM review}, 53\penalty0 (3):\penalty0 526--543, 2011.

\bibitem[Traud et~al.(2012)Traud, Mucha, and Porter]{traud2012social}
A.~L. Traud, P.~J. Mucha, and M.~A. Porter.
\newblock Social structure of {F}acebook networks.
\newblock \emph{Physica A: Statistical Mechanics and its Applications},
  391\penalty0 (16):\penalty0 4165--4180, 2012.

\bibitem[Ugander et~al.(2013)Ugander, Karrer, Backstrom, and
  Kleinberg]{ugander2013graph}
J.~Ugander, B.~Karrer, L.~Backstrom, and J.~Kleinberg.
\newblock Graph cluster randomization: Network exposure to multiple universes.
\newblock In \emph{Proceedings of the 19th ACM SIGKDD International Conference
  on Knowledge Discovery and Data Mining}, pages 329--337. ACM, 2013.

\bibitem[VanderWeele and Robins(2012)]{vanderweele2012stochastic}
T.~J. VanderWeele and J.~M. Robins.
\newblock Stochastic counterfactuals and stochastic sufficient causes.
\newblock \emph{Statistica Sinica}, 22\penalty0 (1):\penalty0 379, 2012.

\bibitem[Wager et~al.(2016)Wager, Du, Taylor, and Tibshirani]{wager2016high}
S.~Wager, W.~Du, J.~Taylor, and R.~J. Tibshirani.
\newblock High-dimensional regression adjustments in randomized experiments.
\newblock \emph{Proceedings of the National Academy of Sciences}, 113\penalty0
  (45):\penalty0 12673--12678, 2016.

\bibitem[Wu(1986)]{wu1986jackknife}
C.-F.~J. Wu.
\newblock Jackknife, bootstrap and other resampling methods in regression
  analysis.
\newblock \emph{The Annals of Statistics}, pages 1261--1295, 1986.

\bibitem[Wu and Gagnon-Bartsch(2017)]{wu2017loop}
E.~Wu and J.~Gagnon-Bartsch.
\newblock The {LOOP} estimator: Adjusting for covariates in randomized
  experiments.
\newblock \emph{arXiv preprint arXiv:1708.01229}, 2017.

\end{thebibliography}


\begin{thebibliography}{79}
\providecommand{\natexlab}[1]{#1}
\providecommand{\url}[1]{\texttt{#1}}
\expandafter\ifx\csname urlstyle\endcsname\relax
  \providecommand{\doi}[1]{doi: #1}\else
  \providecommand{\doi}{doi: \begingroup \urlstyle{rm}\Url}\fi

\bibitem[Abadie et~al.(2017{\natexlab{a}})Abadie, Athey, Imbens, and
  Wooldridge]{abadie2017should}
A.~Abadie, S.~Athey, G.~W. Imbens, and J.~Wooldridge.
\newblock When should you adjust standard errors for clustering?
\newblock Technical report, National Bureau of Economic Research,
  2017{\natexlab{a}}.

\bibitem[Abadie et~al.(2017{\natexlab{b}})Abadie, Athey, Imbens, and
  Wooldridge]{abadie2017sampling}
A.~Abadie, S.~Athey, G.~W. Imbens, and J.~M. Wooldridge.
\newblock Sampling-based vs. design-based uncertainty in regression analysis.
\newblock \emph{arXiv preprint arXiv:1706.01778}, 2017{\natexlab{b}}.

\bibitem[Aral(2016)]{aral2016networked}
S.~Aral.
\newblock \emph{Networked experiments}.
\newblock Oxford, UK: Oxford University Press, 2016.

\bibitem[Aronow(2012)]{aronow2012general}
P.~M. Aronow.
\newblock A general method for detecting interference between units in
  randomized experiments.
\newblock \emph{Sociological Methods \& Research}, 41\penalty0 (1):\penalty0
  3--16, 2012.

\bibitem[Aronow and Middleton(2013)]{aronow2013class}
P.~M. Aronow and J.~A. Middleton.
\newblock A class of unbiased estimators of the average treatment effect in
  randomized experiments.
\newblock \emph{Journal of Causal Inference}, 1\penalty0 (1):\penalty0
  135--154, 2013.

\bibitem[Aronow and Samii(2017)]{aronow2017estimating}
P.~M. Aronow and C.~Samii.
\newblock Estimating average causal effects under general interference, with
  application to a social network experiment.
\newblock \emph{The Annals of Applied Statistics}, 11\penalty0 (4):\penalty0
  1912--1947, 2017.

\bibitem[Athey et~al.(2017{\natexlab{a}})Athey, Eckles, and
  Imbens]{athey2017exact}
S.~Athey, D.~Eckles, and G.~W. Imbens.
\newblock Exact $p$-values for network interference.
\newblock \emph{Journal of the American Statistical Association}, pages 1--11,
  2017{\natexlab{a}}.

\bibitem[Athey et~al.(2017{\natexlab{b}})Athey, Imbens, and
  Wager]{athey2017approximate}
S.~Athey, G.~W. Imbens, and S.~Wager.
\newblock Approximate residual balancing: De-biased inference of average
  treatment effects in high dimensions.
\newblock 2017{\natexlab{b}}.
\newblock URL: \url{https://arxiv.org/pdf/1604.07125.pdf}.

\bibitem[Backstrom et~al.(2012)Backstrom, Boldi, Rosa, Ugander, and
  Vigna]{backstrom2012four}
L.~Backstrom, P.~Boldi, M.~Rosa, J.~Ugander, and S.~Vigna.
\newblock Four degrees of separation.
\newblock In \emph{Proceedings of the 4th Annual ACM Web Science Conference},
  pages 33--42. ACM, 2012.

\bibitem[Baird et~al.(2016)Baird, Bohren, McIntosh, and
  {\"O}zler]{baird2016optimal}
S.~Baird, J.~A. Bohren, C.~McIntosh, and B.~{\"O}zler.
\newblock Optimal design of experiments in the presence of interference.
\newblock \emph{Review of Economics and Statistics}, \penalty0 (0), 2016.

\bibitem[Banerjee et~al.(2013)Banerjee, Chandrasekhar, Duflo, and
  Jackson]{banerjee2013diffusion}
A.~Banerjee, A.~G. Chandrasekhar, E.~Duflo, and M.~O. Jackson.
\newblock The diffusion of microfinance.
\newblock \emph{Science}, 341\penalty0 (6144):\penalty0 1236498, 2013.

\bibitem[Basse and Feller(2018)]{basse2018analyzing}
G.~Basse and A.~Feller.
\newblock Analyzing two-stage experiments in the presence of interference.
\newblock \emph{Journal of the American Statistical Association}, 113\penalty0
  (521):\penalty0 41--55, 2018.

\bibitem[Basse et~al.(2017)Basse, Feller, and Toulis]{basse2017exact}
G.~Basse, A.~Feller, and P.~Toulis.
\newblock Exact tests for two-stage randomized designs in the presence of
  interference.
\newblock \emph{arXiv preprint arXiv:1709.08036}, 2017.

\bibitem[Beaman et~al.(2018)Beaman, BenYishay, Magruder, and
  Mobarak]{beaman2018can}
L.~Beaman, A.~BenYishay, J.~Magruder, and A.~M. Mobarak.
\newblock Can network theory-based targeting increase technology adoption?
\newblock Technical report, National Bureau of Economic Research, 2018.

\bibitem[Berk et~al.(2013)Berk, Pitkin, Brown, Buja, George, and
  Zhao]{berk2013covariance}
R.~Berk, E.~Pitkin, L.~Brown, A.~Buja, E.~George, and L.~Zhao.
\newblock Covariance adjustments for the analysis of randomized field
  experiments.
\newblock \emph{Evaluation Review}, 37\penalty0 (3-4):\penalty0 170--196, 2013.

\bibitem[Bhagat et~al.(2016)Bhagat, Burke, Diuk, Filiz, and Edunov]{bhagat16}
S.~Bhagat, M.~Burke, C.~Diuk, I.~O. Filiz, and S.~Edunov.
\newblock Three and a half degrees of separation.
\newblock \emph{Facebook research note}, 2016.
\newblock URL:
  \url{https://research.fb.com/three-and-a-half-degrees-of-separation/}.

\bibitem[Bloniarz et~al.(2016)Bloniarz, Liu, Zhang, Sekhon, and
  Yu]{bloniarz2016lasso}
A.~Bloniarz, H.~Liu, C.-H. Zhang, J.~S. Sekhon, and B.~Yu.
\newblock Lasso adjustments of treatment effect estimates in randomized
  experiments.
\newblock \emph{Proceedings of the National Academy of Sciences}, 113\penalty0
  (27):\penalty0 7383--7390, 2016.

\bibitem[Blume(1995)]{blume1995statistical}
L.~E. Blume.
\newblock The statistical mechanics of best-response strategy revision.
\newblock \emph{Games and Economic Behavior}, 11\penalty0 (2):\penalty0
  111--145, 1995.

\bibitem[Bramoull{\'e} et~al.(2009)Bramoull{\'e}, Djebbari, and
  Fortin]{bramoulle2009identification}
Y.~Bramoull{\'e}, H.~Djebbari, and B.~Fortin.
\newblock Identification of peer effects through social networks.
\newblock \emph{Journal of Econometrics}, 150\penalty0 (1):\penalty0 41--55,
  2009.

\bibitem[Cai et~al.(2015)Cai, De~Janvry, and Sadoulet]{cai2015social}
J.~Cai, A.~De~Janvry, and E.~Sadoulet.
\newblock Social networks and the decision to insure.
\newblock \emph{American Economic Journal: Applied Economics}, 7\penalty0
  (2):\penalty0 81--108, 2015.

\bibitem[Chen and Lei(2018)]{chen2018network}
K.~Chen and J.~Lei.
\newblock Network cross-validation for determining the number of communities in
  network data.
\newblock \emph{Journal of the American Statistical Association}, 113\penalty0
  (521):\penalty0 241--251, 2018.

\bibitem[Chernozhukov et~al.(2018)Chernozhukov, Chetverikov, Demirer, Duflo,
  Hansen, Newey, and Robins]{chernozhukov2018double}
V.~Chernozhukov, D.~Chetverikov, M.~Demirer, E.~Duflo, C.~Hansen, W.~Newey, and
  J.~Robins.
\newblock Double/debiased machine learning for treatment and structural
  parameters.
\newblock \emph{The Econometrics Journal}, 21\penalty0 (1):\penalty0 C1--C68,
  2018.

\bibitem[Chin(2018)]{chin2018central}
A.~Chin.
\newblock Central limit theorems via {S}tein's method for randomized
  experiments under interference.
\newblock \emph{arXiv preprint arXiv:1808.08683}, 2018.

\bibitem[Choi(2017)]{choi2017estimation}
D.~Choi.
\newblock Estimation of monotone treatment effects in network experiments.
\newblock \emph{Journal of the American Statistical Association}, pages 1--9,
  2017.

\bibitem[Cox(1958)]{cox1958planning}
D.~R. Cox.
\newblock Planning of experiments.
\newblock 1958.

\bibitem[Eckles et~al.(2017)Eckles, Karrer, and Ugander]{eckles2017design}
D.~Eckles, B.~Karrer, and J.~Ugander.
\newblock Design and analysis of experiments in networks: Reducing bias from
  interference.
\newblock \emph{Journal of Causal Inference}, 5\penalty0 (1), 2017.

\bibitem[Efron(1979)]{efron1979bootstrap}
B.~Efron.
\newblock Bootstrap methods: Another look at the jackknife.
\newblock \emph{The Annals of Statistics}, 7\penalty0 (1):\penalty0 1--26,
  1979.

\bibitem[Efron(1987)]{efron1987better}
B.~Efron.
\newblock Better bootstrap confidence intervals.
\newblock \emph{Journal of the American Statistical Association}, 82\penalty0
  (397):\penalty0 171--185, 1987.

\bibitem[Egami(2017)]{egami2017unbiased}
N.~Egami.
\newblock Unbiased estimation and sensitivity analysis for network-specific
  spillover effects: Application to an online network experiment.
\newblock \emph{arXiv preprint arXiv:1708.08171}, 2017.

\bibitem[Eicker(1967)]{eicker1967limit}
F.~Eicker.
\newblock Limit theorems for regressions with unequal and dependent errors.
\newblock 1967.

\bibitem[Forastiere et~al.(2016)Forastiere, Airoldi, and
  Mealli]{forastiere2016identification}
L.~Forastiere, E.~M. Airoldi, and F.~Mealli.
\newblock Identification and estimation of treatment and interference effects
  in observational studies on networks.
\newblock \emph{arXiv preprint arXiv:1609.06245}, 2016.

\bibitem[Fortunato(2010)]{fortunato2010community}
S.~Fortunato.
\newblock Community detection in graphs.
\newblock \emph{Physics reports}, 486\penalty0 (3-5):\penalty0 75--174, 2010.

\bibitem[Freedman(2008{\natexlab{a}})]{freedman2008regressionA}
D.~A. Freedman.
\newblock On regression adjustments to experimental data.
\newblock \emph{Advances in Applied Mathematics}, 40\penalty0 (2):\penalty0
  180--193, 2008{\natexlab{a}}.

\bibitem[Freedman(2008{\natexlab{b}})]{freedman2008regressionB}
D.~A. Freedman.
\newblock On regression adjustments in experiments with several treatments.
\newblock \emph{The annals of applied statistics}, 2\penalty0 (1):\penalty0
  176--196, 2008{\natexlab{b}}.

\bibitem[Fruchterman and Reingold(1991)]{fruchterman1991graph}
T.~M. Fruchterman and E.~M. Reingold.
\newblock Graph drawing by force-directed placement.
\newblock \emph{Software: Practice and Experience}, 21\penalty0 (11):\penalty0
  1129--1164, 1991.

\bibitem[Granovetter(1973)]{granovetter1973strength}
M.~S. Granovetter.
\newblock The strength of weak ties.
\newblock \emph{American Journal of Sociology}, 78\penalty0 (6):\penalty0
  1360--1380, 1973.

\bibitem[Greenland(1987)]{greenland1987interpretation}
S.~Greenland.
\newblock Interpretation and choice of effect measures in epidemiologic
  analyses.
\newblock \emph{American journal of epidemiology}, 125\penalty0 (5):\penalty0
  761--768, 1987.

\bibitem[Hahn(1998)]{hahn1998role}
J.~Hahn.
\newblock On the role of the propensity score in efficient semiparametric
  estimation of average treatment effects.
\newblock \emph{Econometrica}, pages 315--331, 1998.

\bibitem[H{\'a}jek(1971)]{hajek1971comment}
J.~H{\'a}jek.
\newblock Comment on `{A}n essay on the logical foundations of survey sampling,
  part 1' by {D}. {B}asu.
\newblock In V.~Godambe and D.~A. Sprott, editors, \emph{Foundations of
  Statistical Inference}, page 236, Toronto, 1971. Holt, Rinehart and Winston.

\bibitem[Hastie and Tibshirani(1986)]{hastie1986generalized}
T.~Hastie and R.~Tibshirani.
\newblock Generalized additive models.
\newblock \emph{Statistical Science}, 1\penalty0 (3):\penalty0 297--310, 1986.

\bibitem[Haythornthwaite and Wellman(1998)]{haythornthwaite1998work}
C.~Haythornthwaite and B.~Wellman.
\newblock Work, friendship, and media use for information exchange in a
  networked organization.
\newblock \emph{Journal of the American society for information science},
  49\penalty0 (12):\penalty0 1101--1114, 1998.

\bibitem[Horvitz and Thompson(1952)]{horvitz1952generalization}
D.~G. Horvitz and D.~J. Thompson.
\newblock A generalization of sampling without replacement from a finite
  universe.
\newblock \emph{Journal of the American Statistical Association}, 47\penalty0
  (260):\penalty0 663--685, 1952.

\bibitem[Huber(1967)]{huber1967behavior}
P.~J. Huber.
\newblock The behavior of maximum likelihood estimates under nonstandard
  conditions.
\newblock 1967.

\bibitem[Hudgens and Halloran(2008)]{hudgens2008toward}
M.~G. Hudgens and M.~E. Halloran.
\newblock Toward causal inference with interference.
\newblock \emph{Journal of the American Statistical Association}, 103\penalty0
  (482):\penalty0 832--842, 2008.

\bibitem[Imbens(2004)]{imbens2004nonparametric}
G.~W. Imbens.
\newblock Nonparametric estimation of average treatment effects under
  exogeneity: A review.
\newblock \emph{Review of Economics and Statistics}, 86\penalty0 (1):\penalty0
  4--29, 2004.

\bibitem[Jagadeesan et~al.(2017)Jagadeesan, Pillai, and
  Volfovsky]{jagadeesan2017designs}
R.~Jagadeesan, N.~Pillai, and A.~Volfovsky.
\newblock Designs for estimating the treatment effect in networks with
  interference.
\newblock \emph{arXiv preprint arXiv:1705.08524}, 2017.

\bibitem[Kim et~al.(2015)Kim, Hwong, Stafford, Hughes, O'Malley, Fowler, and
  Christakis]{kim2015social}
D.~A. Kim, A.~R. Hwong, D.~Stafford, D.~A. Hughes, A.~J. O'Malley, J.~H.
  Fowler, and N.~A. Christakis.
\newblock Social network targeting to maximise population behaviour change: a
  cluster randomised controlled trial.
\newblock \emph{The Lancet}, 386\penalty0 (9989):\penalty0 145--153, 2015.

\bibitem[Kivel{\"a} et~al.(2014)Kivel{\"a}, Arenas, Barthelemy, Gleeson,
  Moreno, and Porter]{kivela2014multilayer}
M.~Kivel{\"a}, A.~Arenas, M.~Barthelemy, J.~P. Gleeson, Y.~Moreno, and M.~A.
  Porter.
\newblock Multilayer networks.
\newblock \emph{Journal of Complex Networks}, 2\penalty0 (3):\penalty0
  203--271, 2014.

\bibitem[K\"unsch(1986)]{kunsch1989jackknife}
H.~R. K\"unsch.
\newblock The jackknife and the bootstrap for general stationary observations.
\newblock \emph{The Annals of Statistics}, pages 1217--1241, 1989.

\bibitem[Li et~al.(2018)Li, Levina, and Zhu]{li2018network}
T.~Li, E.~Levina, and J.~Zhu.
\newblock Network cross-validation by edge sampling.
\newblock \emph{arXiv preprint arXiv:1612.04717}, 2018.

\bibitem[Lin(2013)]{lin2013agnostic}
W.~Lin.
\newblock Agnostic notes on regression adjustments to experimental data:
  Reexamining {F}reedman’s critique.
\newblock \emph{The Annals of Applied Statistics}, 7\penalty0 (1):\penalty0
  295--318, 2013.

\bibitem[Liu and Hudgens(2014)]{liu2014large}
L.~Liu and M.~G. Hudgens.
\newblock Large sample randomization inference of causal effects in the
  presence of interference.
\newblock \emph{Journal of the American Statistical Association}, 109\penalty0
  (505):\penalty0 288--301, 2014.

\bibitem[Manski(1993)]{manski1993identification}
C.~F. Manski.
\newblock Identification of endogenous social effects: The reflection problem.
\newblock \emph{The Review of Economic Studies}, 1993.

\bibitem[Manski(2013)]{manski2013identification}
C.~F. Manski.
\newblock Identification of treatment response with social interactions.
\newblock \emph{The Econometrics Journal}, 16\penalty0 (1), 2013.

\bibitem[Neyman(1923)]{neyman1923application}
J.~Neyman.
\newblock On the application of probability theory to agricultural experiments.
  {E}ssay on {P}rinciples. section 9. (translated and edited by {D. M.}
  {D}abrowska and {T. P.} {S}peed, \emph{Statistical Science} (1990), 5,
  465-480).
\newblock \emph{Annals of Agricultural Sciences}, 10:\penalty0 1--51, 1923.

\bibitem[Ogburn and VanderWeele(2014)]{ogburn2014causal}
E.~L. Ogburn and T.~J. VanderWeele.
\newblock Causal diagrams for interference.
\newblock \emph{Statistical Science}, 29\penalty0 (4):\penalty0 559--578, 2014.

\bibitem[Paluck et~al.(2016)Paluck, Shepherd, and Aronow]{paluck2016changing}
E.~L. Paluck, H.~Shepherd, and P.~M. Aronow.
\newblock Changing climates of conflict: A social network experiment in 56
  schools.
\newblock \emph{Proceedings of the National Academy of Sciences}, 113\penalty0
  (3):\penalty0 566--571, 2016.

\bibitem[Pearl(2009)]{pearl2009causality}
J.~Pearl.
\newblock \emph{Causality}.
\newblock Cambridge University Press, 2009.

\bibitem[Robins and Greenland(1989)]{robins1989probability}
J.~Robins and S.~Greenland.
\newblock The probability of causation under a stochastic model for individual
  risk.
\newblock \emph{Biometrics}, pages 1125--1138, 1989.

\bibitem[Robins and Greenland(2000)]{robins2000causal}
J.~M. Robins and S.~Greenland.
\newblock Causal inference without counterfactuals: comment.
\newblock \emph{Journal of the American Statistical Association}, 95\penalty0
  (450):\penalty0 431--435, 2000.

\bibitem[Robins et~al.(2000)Robins, Rotnitzky, and
  Scharfstein]{robins2000sensitivity}
J.~M. Robins, A.~Rotnitzky, and D.~O. Scharfstein.
\newblock Sensitivity analysis for selection bias and unmeasured confounding in
  missing data and causal inference models.
\newblock In \emph{Statistical Models in Epidemiology, the Environment, and
  Clinical Trials}, pages 1--94. Springer, 2000.

\bibitem[Rosenbaum(2007)]{rosenbaum2007interference}
P.~R. Rosenbaum.
\newblock Interference between units in randomized experiments.
\newblock \emph{Journal of the American Statistical Association}, 102\penalty0
  (477):\penalty0 191--200, 2007.

\bibitem[Rubin(1974)]{rubin1974estimating}
D.~B. Rubin.
\newblock Estimating causal effects of treatments in randomized and
  nonrandomized studies.
\newblock \emph{Journal of {E}ducational Psychology}, 66\penalty0 (5):\penalty0
  688, 1974.

\bibitem[Rubin(1980)]{rubin1980randomization}
D.~B. Rubin.
\newblock Randomization analysis of experimental data: The {F}isher
  randomization test comment.
\newblock \emph{Journal of the American Statistical Association}, 1980.

\bibitem[Saveski et~al.(2017)Saveski, Pouget-Abadie, Saint-Jacques, Duan,
  Ghosh, Xu, and Airoldi]{saveski2017detecting}
M.~Saveski, J.~Pouget-Abadie, G.~Saint-Jacques, W.~Duan, S.~Ghosh, Y.~Xu, and
  E.~M. Airoldi.
\newblock Detecting network effects: Randomizing over randomized experiments.
\newblock In \emph{Proceedings of the 23rd ACM SIGKDD International Conference
  on Knowledge Discovery and Data Mining}, pages 1027--1035. ACM, 2017.

\bibitem[S{\"a}vje et~al.(2017)S{\"a}vje, Aronow, and
  Hudgens]{savje2017average}
F.~S{\"a}vje, P.~M. Aronow, and M.~G. Hudgens.
\newblock Average treatment effects in the presence of unknown interference.
\newblock \emph{arXiv preprint arXiv:1711.06399}, 2017.

\bibitem[Sussman and Airoldi(2017)]{sussman2017elements}
D.~L. Sussman and E.~M. Airoldi.
\newblock Elements of estimation theory for causal effects in the presence of
  network interference.
\newblock \emph{arXiv preprint arXiv:1702.03578}, 2017.

\bibitem[Taylor and Eckles(2017)]{taylor2017randomized}
S.~J. Taylor and D.~Eckles.
\newblock Randomized experiments to detect and estimate social influence in
  networks.
\newblock \emph{arXiv preprint arXiv:1709.09636}, 2017.

\bibitem[Tchetgen~Tchetgen and VanderWeele(2012)]{tchetgen2012causal}
E.~J. Tchetgen~Tchetgen and T.~J. VanderWeele.
\newblock On causal inference in the presence of interference.
\newblock \emph{Statistical methods in medical research}, 21\penalty0
  (1):\penalty0 55--75, 2012.

\bibitem[Traud et~al.(2011)Traud, Kelsic, Mucha, and
  Porter]{traud2011comparing}
A.~L. Traud, E.~D. Kelsic, P.~J. Mucha, and M.~A. Porter.
\newblock Comparing community structure to characteristics in online collegiate
  social networks.
\newblock \emph{SIAM review}, 53\penalty0 (3):\penalty0 526--543, 2011.

\bibitem[Traud et~al.(2012)Traud, Mucha, and Porter]{traud2012social}
A.~L. Traud, P.~J. Mucha, and M.~A. Porter.
\newblock Social structure of {F}acebook networks.
\newblock \emph{Physica A: Statistical Mechanics and its Applications},
  391\penalty0 (16):\penalty0 4165--4180, 2012.

\bibitem[Tyner et~al.(2017)Tyner, Briatte, and Hofmann]{tyner2017network}
S.~Tyner, F.~Briatte, and H.~Hofmann.
\newblock Network visualization with ggplot2.
\newblock \emph{The R Journal}, 2017.

\bibitem[Ugander and Backstrom(2013)]{ugander2013balanced}
J.~Ugander and L.~Backstrom.
\newblock Balanced label propagation for partitioning massive graphs.
\newblock In \emph{Proceedings of the sixth ACM International Conference on Web
  Search and Data Mining}, pages 507--516. ACM, 2013.

\bibitem[Ugander et~al.(2011)Ugander, Karrer, Backstrom, and
  Marlow]{ugander2011anatomy}
J.~Ugander, B.~Karrer, L.~Backstrom, and C.~Marlow.
\newblock The anatomy of the {F}acebook social graph.
\newblock \emph{arXiv preprint arXiv:1111.4503}, 2011.

\bibitem[Ugander et~al.(2013)Ugander, Karrer, Backstrom, and
  Kleinberg]{ugander2013graph}
J.~Ugander, B.~Karrer, L.~Backstrom, and J.~Kleinberg.
\newblock Graph cluster randomization: Network exposure to multiple universes.
\newblock In \emph{Proceedings of the 19th ACM SIGKDD International Conference
  on Knowledge Discovery and Data Mining}, pages 329--337. ACM, 2013.

\bibitem[van~der Laan(2014)]{van2014causal}
M.~J. van~der Laan.
\newblock Causal inference for a population of causally connected units.
\newblock \emph{Journal of Causal Inference}, 2\penalty0 (1):\penalty0 13--74,
  2014.

\bibitem[VanderWeele and Robins(2012)]{vanderweele2012stochastic}
T.~J. VanderWeele and J.~M. Robins.
\newblock Stochastic counterfactuals and stochastic sufficient causes.
\newblock \emph{Statistica Sinica}, 22\penalty0 (1):\penalty0 379, 2012.


\bibitem[VanderWeele and Tchetgen~Tchetgen(2011)]{vanderweele2011effect}
T.~J. VanderWeele and E.~J. Tchetgen~Tchetgen.
\newblock Effect partitioning under interference in two-stage randomized
  vaccine trials.
\newblock \emph{Statistics \& Probability Letters}, 81\penalty0 (7):\penalty0
  861--869, 2011.

\bibitem[VanderWeele et~al.(2014)VanderWeele, Tchetgen~Tchetgen, and
  Halloran]{vanderweele2014interference}
T.~J. VanderWeele, E.~J. Tchetgen~Tchetgen, and M.~E. Halloran.
\newblock Interference and sensitivity analysis.
\newblock \emph{Statistical Science: A review journal of the Institute of
  Mathematical Statistics}, 29\penalty0 (4):\penalty0 687, 2014.

\bibitem[Wager et~al.(2016)Wager, Du, Taylor, and Tibshirani]{wager2016high}
S.~Wager, W.~Du, J.~Taylor, and R.~J. Tibshirani.
\newblock High-dimensional regression adjustments in randomized experiments.
\newblock \emph{Proceedings of the National Academy of Sciences}, 113\penalty0
  (45):\penalty0 12673--12678, 2016.

\bibitem[Walker and Muchnik(2014)]{walker2014design}
D.~Walker and L.~Muchnik.
\newblock Design of randomized experiments in networks.
\newblock \emph{Proceedings of the IEEE}, 102\penalty0 (12):\penalty0
  1940--1951, 2014.

\bibitem[Watts and Strogatz(1998)]{watts1998collective}
D.~J. Watts and S.~H. Strogatz.
\newblock Collective dynamics of ‘small-world’ networks.
\newblock \emph{Nature}, 393\penalty0 (6684):\penalty0 440, 1998.

\bibitem[White(1980)]{white1980heteroskedasticity}
H.~White.
\newblock A heteroskedasticity-consistent covariance matrix estimator and a
  direct test for heteroskedasticity.
\newblock \emph{Econometrica: Journal of the Econometric Society}, pages
  817--838, 1980.

\bibitem[Wu(1986)]{wu1986jackknife}
C.-F.~J. Wu.
\newblock Jackknife, bootstrap and other resampling methods in regression
  analysis.
\newblock \emph{The Annals of Statistics}, pages 1261--1295, 1986.

\bibitem[Wu and Gagnon-Bartsch(2017)]{wu2017loop}
E.~Wu and J.~Gagnon-Bartsch.
\newblock The {LOOP} estimator: Adjusting for covariates in randomized
  experiments.
\newblock \emph{arXiv preprint arXiv:1708.01229}, 2017.

\end{thebibliography}

\pagebreak
\appendix

\section{Proofs for Section~\ref{sec:linear}}

\subsection{Proof of Proposition~\ref{prop:linear-unbiased}}
\begin{proof}
Let $\ep_w$ be the $N_w$ vector of $w$-group residuals.  As $y_w = X_w \beta_w + \ep_w$, for $w = 0, 1$, conditionally on $X_w$ being full rank we have
\[\E[\hat \beta_w] = \E[(X_w^\top X_w)^{-1}X_w^\top y_w] = \E[(X_w^\top X_w)^{-1} X_w^\top (X_w \beta_w + \ep_w)] = \beta_w + \E[(X_w^\top X_w)^{-1} X_w^\top \ep_w].\]
Assumption~\ref{asm:errors}(\ref{asm:errors-exogenous}) ensures that the second term is zero, and thus $\hat \beta_w$ is unbiased for $\beta_w$.  

Unbiasedness of $\hat \tau$ then follows by linearity of expectation.
\end{proof}

\subsection{Proof of Theorem~\ref{thm:ols-variance-finite}}
\begin{proof}
We first calculate the variance of $\hat \beta_w$.  By the law of total variance, we have
\begin{align*}
\var(\hat \beta_w) &= \var[(X_w^\top  X_w)^{-1}  X_w^\top y_w] \\
&= \var[( X_w^\top X_w)^{-1}  X_w^\top \ep_w] \\
&= \E[(X_w^\top  X_w)^{-1} X_w^\top \var(\ep_w | X_w) X_w( X_w^\top  X_w)^{-1}] + \var[( X_w^\top X_w)^{-1}  X_w^\top \E(\ep_w | X_w)].
\end{align*}
The second term is equal to zero by Assumption~\ref{asm:errors}(\ref{asm:errors-exogenous}), and so by Assumption~\ref{asm:errors}(\ref{asm:errors-independent}) and~(\ref{asm:errors-homoscedastic}),
\[\var(\hat \beta_w) = \sigma^2\E[(X_w^\top X_w)^{-1}].\]

The coefficient estimates of the two groups are uncorrelated because the residuals are uncorrelated.  That is,
\begin{align*}
\cov(\hat \beta_0, \hat \beta_1) &= \E[\cov(\hat \beta_0, \hat \beta_1 | X)] + \cov(\E[\hat \beta_0 | X], \E[ \hat \beta_1 | X]) \\
&= \E[\cov((X_0^\top X_0)^{-1} X_0^\top \ep_0, (X_1^\top X_1)^{-1} X_1^\top \ep_1)] + 0 \\
&= 0.
\end{align*}

Therefore, 
\begin{align*}
\var(\hat \tau) &= \var((\omega_1)^\top \hat \beta_1 - (\omega_0)^\top \hat \beta_0) \\
&= \sigma^2\left((\omega_0)^\top \E[(X_0^\top X_0)^{-1}] \omega_0 + (\omega_1)^\top \E[(X_1^\top X_1)^{-1}] \omega_1\right),
\end{align*}
which produces the variance expression in equation~\eqref{eqn:ols-variance-finite}.
\end{proof}

\subsection{Proof of Theorem~\ref{thm:ols-clt}}
This lemma establishes some basic convergence results.
\begin{lemma}
\label{lem:convergences}
Let $\bar X_w$ and $S_w$ denote the within-group sample means and covariances.
Under Assumptions~\ref{asm:bernoulli-design}, \ref{asm:x-indep-w}, and the assumptions in the statement of Theorem~\ref{thm:ols-clt}, the following statements hold for $w = 0, 1$.
\begin{enumerate}[(a)]
\item $\bar X_w \pto \mu_X$. 
\item $S_w \pto \Sigma_X$.
\item $\hat \eta_w \pto \eta_w$.
\item $\sqrt{n\pi}(\bar X_1 - \mu_X) \wto N(0, \Sigma_X)$ and $\sqrt{n(1 - \pi)}(\bar X_0 - \mu_X) \wto N(0, \Sigma_X)$.
\item $\sqrt{n\pi}(\hat \eta_1 - \eta_1) \wto N(0, \sigma^2\Sigma_X^{-1})$ and $\sqrt{n(1 - \pi)}(\hat \eta_0 - \eta_0) \wto N(0, \sigma^2 \Sigma_X^{-1})$.
\item $\sqrt{n}(\bar \ep_1 - \bar \ep_0) \wto N\left(0, \frac{\sigma^2}{\pi(1 - \pi)}\right)$.
\end{enumerate}
\end{lemma}
\begin{proof}{}
\begin{enumerate}[(a)]
\item Because of Bernoulli random sampling it holds that
\[\limn \E[\bar X_1] = \limn \E\left[\frac{1}{N_1}\sumin W_i X_i\right] = \mu_X.\]

By conditioning on $X$ we have
\[\var(\bar X_1) = \E[\var (\bar X_1 | X)] + \var[\E(\bar X_1 | X)].\]
For the first term, we have
\[\E[\var(\bar X_1 | X)] = \E\left[\var\left(\frac{1}{n\pi} \sumin W_i X_i + r_n\right)\right],\]
where
\[\var\left(\frac{1}{n\pi} \sumin W_i X_i\right) = \frac{1 - \pi}{n^2\pi}  \sumin X_i^2 = O_p(n^{-1})\]
and
\[r_n = \left(\frac{1}{N_1} - \frac{1}{np} \right)\sumin W_i X_i = O_p(n^{-1})\]
since $N_1/n \to \pi$ in probability.
For the second term, we have
\[\var[\E(\bar X_1 | X] = \var(\bar X ) \to 0\]
since $\bar X - \mu_X = o_p(1)$.
Therefore, we conclude $\var (\bar X_1) \to 0$, and so consistency follows from Chebychev's inequality. 

The result similarly holds for $\bar X_0$.
\item This result is established in a similar manner to part (a), using the fact that
\[\avgin (X_i - \bar X)^\top (X_i - \bar X) \pto \Sigma_X,\]
and the fact that fourth moments are bounded.
\item The convergence of $\hat \eta_w$ to $\eta_w$ follows conditionally on $X$ from standard OLS theory.  Then, letting
\[S_w = \frac{1}{n}(X_w - \bar X_w)^\top (X_w - \bar X_w)\]
denote the sample covariance matrix,
we find
\begin{align*}
\var(\hat \eta) &= \var[\E[\hat \eta_w | X]] + \E[\var[\hat \eta_w | X]] \\
&= \var[\eta_w] + \frac{\sigma^2}{n}\E[S_w^{-1}] \to 0.
\end{align*}
Convergence in probability follows from Chebychev's inequality.
\item This result follows from Bernoulli sampling and the convergence of the finite population means, $\bar X \pto \mu_X$.
\item As in the proof of part (c), we write
\[\hat \eta_w = \frac{1}{n}S_w^{-1}(X_w - \bar X_w)^\top (y_w - \bar y_w).\]
Since $y_w = X_w\eta_w + \ep_w$, we can write
\begin{align*}
\sqrt{n}(\hat \eta_w - \eta_w) &= \sqrt{n}\left[\frac{1}{n}S_w^{-1} (X_w - \bar X_w)^\top (y_w - \bar y_w)  - \eta_w\right]\\
&= \frac{1}{\sqrt{n}}S_w^{-1} (X_w - \bar X_w)^\top (\ep_w - \bar \ep_w) \\
&= \frac{1}{\sqrt{n}}\Sigma_X^{-1} (X_w - \bar X_w)^\top (\ep_w - \bar \ep_w) + R,
\end{align*}
where the remainder is
\[R = \frac{1}{\sqrt{n}}(S_w^{-1} - \Sigma_X^{-1}) (X_w - \bar X_w)^\top (\ep_w - \bar \ep_w)\]
Since $S_w^{-1}- \Sigma_X^{-1} = o_p(1)$ is implied by $S_w \pto\Sigma_X$, and $\sqrt{n}(X_w - \bar X_w)^\top = O_p(1)$ and $\sqrt{n}(\ep_w - \bar \ep_w) = O_p(1)$, the remainder satisfies $R = o_p(1)$.

Then $\sqrt{n}(\hat \eta_w - \eta_w)$ is asymptotically Gaussian with mean zero and variance
\[\limn \var\left(\frac{1}{\sqrt{n}}\Sigma_X^{-1}(X_w - \bar X_w)^\top (\ep_w - \bar \ep_w)\right) = \sigma^2 \Sigma_X^{-1} \limn \var(X_w) \Sigma_X^{-1}.\]
Using the result of part (d), this variance equals $\frac{\sigma^2}{\pi} \Sigma_X^{-1} \Sigma_X \Sigma_X^{-1} = \frac{\sigma^2}{\pi} \Sigma_X^{-1}$ when $w = 1$ and $\frac{\sigma^2}{1 - \pi}\Sigma_X^{-1}$ when $w = 0$.
\item From Assumption~\ref{asm:errors}, $\bar \ep_1$ is independent of $\bar \ep_0$ with variances $\sigma^2 / (n \pi)$ and $\sigma^2 / (n(1 - \pi))$, respectively.  A standard central limit theorem shows that $\sqrt{n} (\bar \ep_1 - \bar \ep_0)$ is asymptotically Gaussian with mean $0$ and variance
\[\frac{\sigma^2}{\pi} + \frac{\sigma^2}{1 - \pi} = \frac{\sigma^2}{\pi(1 - \pi)}.\]
\end{enumerate}
\end{proof}

We now prove the main theorem.
 
\begin{proof}
We characterize the treatment effect estimator as
\begin{align*}
\hat \tau - \tau &= \bar y_1 - \bar y_0 + (\omega_1 - \bar X_1)^\top \hat \eta_1 - (\omega_0 - \bar X_0)^\top \hat \eta_0 - (\alpha_1 - \alpha_0) - (\omega_1^\top \eta_1 - \omega_0^\top \eta_0) \\
&= \bar \ep_1 - \bar \ep_0 + (\omega_1 - \bar X_1)^\top (\hat \eta_1 - \eta_1) -  (\omega_0 - \bar X_0)^\top (\hat \eta_0 - \eta_0),  
\end{align*}
which implies that
\begin{align*}
\sqrt{n}(\hat\tau - \tau)&= \sqrt{n}(\bar \ep_1 - \bar \ep_0) + \sqrt{n}(\omega_1 - \bar X_1)^\top (\hat \eta_1 - \eta_1) - \sqrt{n}(\omega_0 - \bar X_0)^\top (\hat \eta_0 - \eta_0).
\end{align*}
Now,
\[\sqrt{n}(\omega_w - \bar X_w)^\top (\hat \eta_w - \eta_w) = \sqrt{n}(\omega_w - \mu_w)^\top (\hat \eta_w - \eta_w) + \sqrt{n}(\mu_w - \bar X_w)^\top (\hat \eta_w - \eta_w),\]
for $w = 0, 1$, where the second term is $o_p(1)$ since $\bar X_w \pto \mu_X$ and $\hat \eta_w \pto \eta_w$ following from parts (a) and (c) of Lemma~\ref{lem:convergences}.  Therefore,
\[\sqrt{n}(\hat \tau - \tau) = \sqrt{n}(\bar \ep_1 - \bar \ep_0) + \sqrt{n}(\omega_1 - \mu_X)^\top (\hat \eta_1 - \eta_1) - \sqrt{n}(\omega_0 - \mu_X)^\top (\hat \eta_0 - \eta_0) + o_p(1).
\]

The three terms are uncorrelated, with
\begin{align*}
\sqrt{n}(\bar \ep_1 - \bar \ep_0)& \wto N\left(0, \frac{\sigma^2}{\pi(1 - \pi)}\right) \\
\sqrt{n}(\omega_1 - \mu_X)^\top (\hat \eta_1 - \eta_1) &\wto N\left(0, \frac{\sigma^2}{\pi}\|\omega_1 - \mu\|_{\Sigma_X^{-1}}^2\right) \\
\sqrt{n}(\omega_0 - \mu_X)^\top (\hat \eta_0 - \eta_0) &\wto N\left(0, \frac{\sigma^2}{1 - \pi} \|\omega_0 - \mu\|_{\Sigma_X^{-1}}^2\right),
\end{align*}
established in parts (e) and (f) of Lemma~\ref{lem:convergences}.
Combining the terms produces the variance expression in equation~\eqref{eqn:asymp-var}, and completes the proof.
\end{proof}

\subsection{Proof of Corollary~\ref{cor:ols-fixed}}
\begin{proof}
If $X_i$ is independent of $\Wni$, then 
\[\omega_0 = \avgin \E[X_i | \Wni = \+0] = \avgin \E[X_i],\]
and so is equal to $\mu_X$ in the limit (with the understanding that $\omega_0$ is actually a sequence associated with each finite population).  The same holds true for $\omega_1$.  Then the result follows immediately from equation~\eqref{eqn:asymp-var}, as the second and third terms are equal to zero.
\end{proof}


\end{document}